\documentclass[10pt,conference]{IEEEtran}
\IEEEoverridecommandlockouts

\usepackage{cite}
\usepackage{amsmath}
\usepackage{amsfonts,amsmath,amssymb,bm,amsthm}
\usepackage{amsthm}
\usepackage{algorithmic}
\usepackage{graphicx}
\usepackage{textcomp}
\usepackage{xcolor}
\usepackage{braket}
\usepackage[hidelinks]{hyperref}
\usepackage[T1]{fontenc}
\usepackage{cleveref}
\usepackage{multirow}
\usepackage{booktabs}
\usepackage{adjustbox}
\usepackage{makecell}
\usepackage{pdflscape}
\usepackage{placeins}
\usepackage{multirow}
\usepackage{subcaption}
\IfFileExists{pdf-colors.tex}{\input{pdf-colors}}{}
\crefformat{equation}{(#2)#1(#3)}
\crefrangeformat{equation}{(#3#1#4)--(#5#2#6)}

\def\BibTeX{{\rm B\kern-.05em{\sc i\kern-.025em b}\kern-.08em
    T\kern-.1667em\lower.7ex\hbox{E}\kern-.125emX}}

\DeclareMathOperator{\tr}{tr}

\newtheorem{lemma}{Lemma}
\newtheorem{corollary}{Corollary}

\newtheorem{definition}{Definition}
\newtheorem{criterion}{Criterion}
\newtheorem{remark}{Remark}

\allowdisplaybreaks

\renewcommand{\thesubsubsection}{\arabic{subsubsection}}

\makeatletter 
\newcommand{\linebreakand}{%
  \end{@IEEEauthorhalign}
  \hfill\mbox{}\par
  \mbox{}\hfill\begin{@IEEEauthorhalign}
}
\makeatother 

\renewcommand{\thesubsubsection}{\arabic{subsubsection}}

\begin{document}

\title{Toward Covert Quantum Computing
\thanks{This work was supported by the National Science Foundation under Grant No.~CNS-2107265 and Noblis, Inc. We acknowledge the use of IBM Quantum Credits for this work. The views expressed are those of the authors, and do not reflect the official policy or position of IBM or the IBM Quantum team.}}

\author{\IEEEauthorblockN{Evan J. D. Anderson}
\IEEEauthorblockA{\textit{Wyant College of Optical Sciences} \\
\textit{University of Arizona}\\
Tucson, AZ, USA \\
ejdanderson@arizona.edu}
\and
\IEEEauthorblockN{Kaushik Datta}
\IEEEauthorblockA{\textit{Noblis, Inc.} \\
Reston, VA, USA\\
Kaushik.Datta@noblis.org}
\and
\IEEEauthorblockN{Boulat A. Bash}
\IEEEauthorblockA{\textit{Dept. of Elec. and Comp. Eng.} \\
\textit{College of Engineering} \\
\textit{Wyant College of Optical Sciences} \\
\textit{University of Arizona}\\
Tucson, AZ, USA\\
boulat@arizona.edu}
}

\maketitle

\begin{abstract}
As quantum computers become available through multi-tenant cloud platforms, ensuring privacy against adversaries sharing the same quantum processing unit becomes critical. We introduce and explore \emph{covert quantum computing}, a new concept that ensures an adversary with access to all other quantum computational units (QCUs) of a quantum computer cannot detect computation on the subset that they cannot access. Analogous to covert communication, we employ information theory. However, since here the adversary controls the systems used for detection, we require a richer framework for covertness analysis that accounts for the use of quantum memories and adaptive operations. Thus, we adopt the \emph{quantum-strategy} framework used in quantum game theory and memory channel discrimination. Current quantum computers use planar graph circuit layouts and typically assume nearest-neighbor crosstalk.  We derive discrete isoperimetric inequalities to show that, for an $n$-qubit circuit under this model, only $\mathcal{O}(\sqrt{n})$ border qubits provide detection information to the adversary. We then explore this scaling law on IQM's 54-qubit \emph{Emerald} processor and IBM's 156-qubit \emph{ibm\_fez} machine employing the Heron 2 architecture. We implement Ramsey experiments on qubits not used in computation, and detect nearest-neighbor crosstalk, as expected. However, we also observe long-range coupling effects beyond the border qubits, revealing a side channel that the adversary can exploit. We hypothesize that this long-range crosstalk is induced by leakage from the drive and control lines. Beyond weakening covertness, it exposes co-tenants to both adversarial and unintended crosstalk and degrades circuits that span spatially distributed qubits, motivating further work on spatial isolation and crosstalk characterization.
\end{abstract}


\section{Introduction}

The resource costs of building and operating quantum computers suggest that cloud-based access will remain the dominant model for the foreseeable future. As these systems grow in qubit counts and improve gate fidelities, a natural consequence is multi-tenancy: multiple users executing quantum circuits simultaneously on disjoint subsets of the same quantum processing unit (QPU), a practice already standard in classical cloud computing~\cite{chenNISQQuantumComputing2024, xuSecuringQuantumComputer2025, volyaSecureClassicalQuantumSystems2023}. In this setting, \emph{covert quantum computing} poses a question that has yet to be directly explored: can one user detect that another user is computing on the same QPU?

In the classical setting, private computation \cite{evansPragmaticIntroductionSecure2018}, including secure multi-party computation and homomorphic encryption~\cite{gentryFullyHomomorphicEncryption2009}, protects the inputs, algorithms, and outputs of a computation, but does not conceal that computation is occurring. Blind, or private, quantum computing~\cite{broadbentUniversalBlindQuantum2009, fitzsimonsPrivateQuantumComputation2017a} is the quantum analogue, preventing an untrusted QPU provider from learning the algorithm or its results. Similarly, verifiable quantum computing~\cite{mahadevClassicalVerificationQuantum2022} allows a classical client to confirm the provider executed the requested quantum computation rather than, say, a classical approximation. In the multi-tenant setting, existing literature on QPUs focuses on whether an adversary can learn another user's circuit or results~\cite{xuSecuringQuantumComputer2025, tanKnowWhatYou2025} or corrupt their computation~\cite{arellanoQubitViseDoubleSidedCrosstalk2025, campbellSchrodingersToolboxExploring2025, harperCrosstalkAttacksDefence2025}. In these frameworks, the adversary is assumed to already know that computation is taking place.

The idea of covertness, or low probability of detection (LPD), was first established for classical communication in \cite{bashSquareRootLaw2012, bashLimitsReliableCommunication2013} and has had many follow-on works \cite{bashCovertOpticalCommunication2015, bloch15covert, wang15covert, bashHidingInformationNoise2015, gagatsosCovertCapacityBosonic2020, tahmasbi21signalingcovert, bullockFundamentalLimitsCovert2025,  zlotnick25eacqcovert,
kimelfeldCovertEntanglementGeneration2026, andersonFundamentalLimitsBosonic2021}; see \cite{chen23covcommssurvey} for a comprehensive survey. Recently, its characterization was extended to covert quantum communication \cite{anderson2025covert-isit, anderson2025CovertEntanglementJSAC, andersonAchievabilityCovertQuantum2025, andersonSquareRootLaw2024}.

In either quantum or classical covert communication, Alice attempts to transmit to Bob, while an adversary, Warden Willie,\footnote{Warden Willie moniker is borrowed from the steganography literature \cite{fridrich09stego}.} attempts to detect whether Alice transmitted.
The square-root law (SRL) governs such channels: $M(n) \propto \sqrt{n}$ bits or qubits can be reliably and covertly transmitted from Alice to Bob in $n$ channel uses.

In contrast to covert communication, where Willie's default state corresponding to Alice not transmitting is fixed by the noisy channel, covert quantum computing assumes Willie can fully control the quantum computational units (QCUs) (qubits, qudits, or qumodes) he uses to detect Alice's computation, and has access to an additional reference system and quantum memory. Thus, we require a different framework to characterize Willie's full capabilities. We adopt the quantum-strategy framework from
\cite{gutoskiGeneralTheoryQuantum2007, gutoskiMeasureDistanceQuantum2012} which is operationally equivalent to the \emph{quantum comb} and \emph{tester} framework in \cite{chiribellaMemoryEffectsQuantum2008, chiribellaTheoreticalFrameworkQuantum2009} for memory channels. This permits the description of Willie's detection strategies to include operations on both his QCUs and the reference system while Alice is performing her computation.
The present work extends Anderson's dissertation \cite{andersonCovertQuantumCommunication2025}, which formulated covert quantum computing from the more restrictive prepare-wait-and-measure point of view, where Willie prepares an initial joint state on his QCUs, idles during Alice's computation, and performs a final measurement.

Beyond the multi-tenant QPU security and covert communication, two additional results inspired our investigation of covert quantum computing. In covert cycle stealing, a malicious user attempts to add jobs to a server queue without detection.  An SRL similar to that for covert communication dictates that $\mathcal{O}(\sqrt{n})$ jobs can be added over $n$ processing periods with negligible change in job arrival statistics~\cite{jiangCovertCycleStealing2021, yardiCovertQueueingProblem2021}. However, the mechanisms underlying the covert communication and cycle-stealing bounds do not translate directly to quantum computing. A closer parallel comes from outside the covertness literature: in a wireless network with $n$ optimally placed nodes and unit area, the transport capacity scales as $\Theta(\sqrt{n})$ bit meters per second \cite{guptaCapacityWirelessNetworks2000}. Subsequent work demonstrated that the $\sqrt{n}$ scaling is fundamentally due to the number of spatial degrees of freedom crossing any cut in a planar network being determined by the length of the cut, which scales as $\sqrt{n}$ in two dimensions. The scaling is thus intrinsic to the planar geometry and is not detection or protocol dependent \cite{franceschettiCapacityWirelessNetworks2009a}.
Indeed, a similar geometric argument underlies the square-root law for the number of idling qubits that Alice needs to prevent the detection of her computation from Willie, which we derive.

Our contributions are threefold. First, we formulate covert quantum computing as an information-theory problem, providing a notion of reliable quantum computing and casting Willie's detection task as a quantum strategy discrimination between idling (null hypothesis) and an active computation (alternate hypothesis). This allows us to provide a norm-based criterion for $\delta$-covertness that accounts for an optimal adversary with arbitrary quantum memory and adaptive strategies between operations of Alice's computation.

Second, we specialize the framework above to superconducting QPUs under a nearest-neighbor crosstalk model. Developing discrete vertex-isoperimetric inequalities for the heavy-hex lattice by edge subdivision of the hexagonal lattice, and combining with a known bound on the square lattice, we prove that any $n$-qubit circuit admits a placement whose boundary has size $\mathcal{O}(\sqrt{n})$. Consequently, Alice achieves covertness with only $\mathcal{O}(\sqrt{n})$ additional idling qubits that form a physical boundary around her computation.

Third, we experimentally validate our square-root scaling on IQM's 54-qubit \emph{Emerald} (square lattice) and IBM's 156-qubit \emph{ibm\_fez} (heavy-hex lattice) processors using Ramsey interferometry on spectator qubits. Nearest-neighbor crosstalk is detectable on both devices, consistent with the theoretical model. We additionally observe long-range couplings beyond the nearest-neighbor boundary in both architectures, revealing a side channel that extends Willie's detection radius beyond the isoperimetric prediction, and motivating revised buffer constructions for covert placement.

The rest of the paper is organized as follows: in Section~\ref{sec:covert-quantum-computing-theory}, we provide an information-theoretic description of covert quantum computing. In Section~\ref{sec:scaling-bound} we describe the tools required to explore covert quantum computing on current superconducting QPUs and develop a square-root bound for the number of additional qubits required for covertness. Next, in Section~\ref{sec:experiments} we present our experimental results on IQM's \emph{Emerald} and IBM's \emph{ibm\_fez} QPUs. Finally, we close in Section~\ref{sec:conclusion} with a discussion of our results and future research directions.

\section{Covert Quantum Computing}
\label{sec:covert-quantum-computing-theory}

\subsection{Diamond-Norm Distance}
\label{sec:prelim-diamond-distance}
The \emph{diamond-norm distance} quantifies the distinguishability of quantum channels. It is a norm on the space of quantum channels and is defined as the maximum trace-norm distance between channel outputs over all input states, including those entangled with a reference system $R$~\cite[Def.~9.1.3]{wilde16quantumit2ed}: $\|\mathcal{N}-\mathcal{N}^\prime\|_{\diamond} \equiv \max _{\hat{\rho}_{R A}}\left\|\left(\hat{I}_{R} \otimes \mathcal{N}_{A \rightarrow B}\right)\left(\hat{\rho}_{R A}\right)-\left(\hat{I}_{R} \otimes \mathcal{N}^\prime_{A \rightarrow B}\right)\left(\hat{\rho}_{R A}\right)\right\|_1$. Operationally, given a single use of an unknown channel that is either $\mathcal{N}$ with prior $p$ or $\mathcal{N}^\prime$ with prior $1-p$, the maximum probability of successfully identifying the channel is $p_{succ}=\tfrac{1}{2}+\tfrac{1}{2}\|p\mathcal{N}-(1-p)\mathcal{N}^\prime\|_{\diamond}$. As such, the diamond-norm distance may be defined with the $\frac{1}{2}$ prefactor. Thus, the diamond-norm distance is the channel analogue of trace distance for state discrimination.

\subsection{Reliability}
Analogous to covert communication, covert quantum computing without reliability is trivial. Alice can be undetectable by simply idling.
We, therefore, require that Alice's computation is first reliable, that is, her physical implementation on noisy systems approximates the target computation within an acceptable error.
We use quantum computational units (QCUs) as the basic resources for quantum computation (i.e., qubits, qudits, and qumodes) and formalize reliability using quantum instruments \cite{daviesOperationalApproachQuantum1970, ozawaQuantumMeasuringProcesses1984}\cite[Sec.~4.6.8]{wilde16quantumit2ed} and the diamond-norm distance defined in Section \ref{sec:prelim-diamond-distance}. Quantum instruments generalize unitary evolutions and completely positive trace-preserving (CPTP) maps to include measurement and conditional operations. This generality captures computational tasks, such as quantum teleportation protocols that require mid-circuit measurement and classical feedback \cite[Sec.~7.4]{wilde16quantumit2ed}.

\begin{definition}[Reliable Quantum Computation]\label{def:reliable-quantum-computation}
Let $Z$ be a classical register with orthonormal basis $\{\ket{z}\}_{z\in\mathcal{Z}}$, and let $\mathcal{I}=\{\mathcal{E}_z\}_{z\in\mathcal{Z}}$ be a quantum instrument representing the target computation, where each $\mathcal{E}_z$ is a completely positive, trace-non-increasing map corresponding to measurement outcome $z$, and the sum $\sum_{z\in\mathcal{Z}}\mathcal{E}_z$ is trace-preserving.

Let $\tilde{\mathcal{I}}=\left\{\tilde{\mathcal{E}}_z\right\}_{z\in\mathcal{Z}}$ denote an implementation of $\mathcal{I}$. Each instrument induces a map that outputs both the quantum state after operations and the classical outcome:
\begin{align}
\mathcal{N}^{{\mathcal{I}}}(\hat{\rho}) &\equiv \sum_{z\in\mathcal{Z}} \ket{z}\bra{z}_Z \otimes \mathcal{E}_z(\hat{\rho}),\\
\mathcal{N}^{\tilde{\mathcal{I}}}(\hat{\rho}) &\equiv \sum_{z\in\mathcal{Z}} \ket{z}\bra{z}_Z \otimes \tilde{\mathcal{E}}_z(\hat{\rho}).
\end{align}
An implementation $\tilde{\mathcal{I}}$ is $(n,\epsilon)$-reliable if it uses at most $n$ quantum computational units and
\begin{align}
\|\mathcal{N}^{{\mathcal{I}}} - \mathcal{N}^{\tilde{\mathcal{I}}}\|_\diamond \le \epsilon. \label{eq:reliability-diamond-distance}
\end{align}
\end{definition}

We use the diamond-norm distance to capture the worst-case error across all inputs, including those entangled with external reference systems. This is consistent with rigorous treatments of fault-tolerant quantum computation, where threshold theorems guarantee achievable fault tolerance for physical error rates below a constant threshold~\cite[Sec.~10.6.1]{nielsen11QuantumCompuingandQuantumInfo10th}\cite{kitaevQuantumComputationsAlgorithms1997, aharonovQuantumCircuitsMixed1998, hashimBenchmarkingQuantumLogic2023}.

\subsection{Covertness}
\label{sec:covertness}
Consider a quantum system with $N$ total QCUs. Suppose Alice can use a subset of size $n<N$ to perform her computation. We analyze Alice's covertness against the strongest possible adversary. Thus, we allow Willie access to $s = N - n$ QCUs not allocated to Alice, denoted as the joint system $W^s$, and a quantum memory $M$ of arbitrary dimension. Willie must decide between two hypotheses: $H_0$ where Alice is idle, and $H_1$ corresponding to Alice implementing $\tilde{\mathcal{I}}$.

We adopt the quantum-strategy framework \cite{gutoskiGeneralTheoryQuantum2007, gutoskiMeasureDistanceQuantum2012} (see also quantum combs and testers \cite{chiribellaMemoryEffectsQuantum2008, chiribellaTheoreticalFrameworkQuantum2009}) to characterize Willie's discrimination capabilities. In this framework, the physical executions of Alice's implementation $\tilde{\mathcal{I}}$ induce an $r$-\emph{round} strategy $S_A^{(1)}(\tilde{\mathcal{I}})$ on Willie's observations. If Alice is not computing, the process is an idle strategy $S_A^{(0)}$. 

Mathematically, an $r$-round strategy consists of $r$ ordered input-output interaction slots. Each slot has an incoming system and outgoing system, which are interpreted as \emph{messages} in this formalism. In our setting, they represent QPU-based coupling between Alice's and Willie's QCUs rather than intentional messages. Alice's strategy and Willie's co-strategy are each realized by causally ordered sequences of CPTP maps connected by a quantum memory,  equivalently represented by positive semidefinite operators satisfying causality constraints. The value of $r$ is determined by the timing of Willie's allowed interactions on $W^sM$ and not necessarily by the number of operations Alice uses for $\tilde{\mathcal{I}}$.

Willie uses an $r$-round measuring co-strategy $S_W$, the compatible opposing process to Alice's strategy with a final decision measurement, to decide between $H_0$ and $H_1$. To grant Willie maximum advantage, we assume he possesses full knowledge of $S_A^{(1)}(\tilde{\mathcal{I}})$, and Alice's individual operations chosen for implementing $\tilde{\mathcal{I}}$. Willie's detection performance is characterized by the probabilities $P_{\rm FA} = \Pr(\text{choose } H_1 | H_0)$ and $P_{\rm MD} = \Pr(\text{choose } H_0 \mid H_1)$ of false alarm and missed detection events, respectively. We assume equal prior probabilities for the hypotheses $P(H_0) = P(H_1)=\frac{1}{2}$, though this can be relaxed (see, e.g., \cite{sobersCovertCommunicationPresence2017}). Willie's total error probability is thus $P_W^{\text{(e)}}(S_W,\tilde{\mathcal{I}}) = \tfrac{1}{2}(P_{\rm FA} + P_{\rm MD})$. This leads to our definition of covertness in quantum computation:

\begin{definition}[Covert Quantum Computation]
\label{def:covert-qc}
An implementation $\tilde{\mathcal{I}}$ is $\delta$-\emph{covert} if, for any adversary strategy $S_W$, $P_W^{\mathrm{(e)}}(S_W,\tilde{\mathcal{I}}) \ge \tfrac{1}{2} - \delta$. It is \emph{covert} if it is $\delta$-covert for every $\delta > 0$.
\end{definition}

The \emph{strategy $r$-norm} quantifies the distinguishability of $r$-round strategies as the trace and diamond norms do for quantum states and channels, respectively.
It is the supremum is over all non-measuring co-strategies followed by the optimal Helstrom measurement:
\begin{align}
    \|S_0-S_1\|_{\diamond r}\equiv\sup_{S_W}\|\hat{\rho}_1(S_W) - \hat{\rho}_0(S_W)\|_1,
\end{align}
where $\hat{\rho}_1(S_W)$ and $ \hat{\rho}_0(S_W)$ are the states  at the end of $r$-round strategies $S_0$ and $S_1$. Therefore, Willie's optimal probability of error in distinguishing $S_A^{(0)}$ from $S_A^{(1)}(\tilde{\mathcal{I}})$ is \cite[Sec.~1.2]{gutoskiMeasureDistanceQuantum2012}:
\begin{align}
\min P^{\text{(e)}}_W &\equiv \inf_{S_W} P^{\text{(e)}}_W(S_W)\\&= \frac{1}{2} - \frac{1}{4}
\left\| S_A^{(1)}(\tilde{\mathcal{I}}) - S_A^{(0)} \right\|_{\diamond r} \label{eq:strategy-norm} \\
&=\frac{1}{2}-\frac{1}{4}\sup_{S_W}\|\hat{\rho}^{(1)}_{W^s M}(S_W,\tilde{\mathcal{I}})-\hat{\rho}^{(0)}_{W^s M}(S_W)\|_1.
\end{align}
The states $\hat{\rho}^{(0)}_{W^s M}(S_W)$ and $\hat{\rho}^{(1)}_{W^s M}(S_W,\tilde{\mathcal{I}})$ represent Willie's final joint state under the idle and computing hypotheses. In the remainder, we drop explicit dependence on $S_W$ and $\tilde{\mathcal{I}}$ when it is clear from the context. This yields the following covertness criterion:
\begin{criterion}\label{crit:covert-strategy-computation}
$\tilde{\mathcal{I}}$ is $\delta$-covert if and only if:
\begin{align}
\frac{1}{4}\left\| S_A^{(1)}(\tilde{\mathcal{I}}) - S_A^{(0)} \right\|_{\diamond r} \le \delta.
\end{align}
\end{criterion}

\begin{remark}[Unknown $\tilde{\mathcal{I}}$]\label{rem:unknown-i}
If Willie does not know $\tilde{\mathcal{I}}$ nor the specific operations Alice uses to implement it, then the alternate hypothesis is a composite one. Let $\mathbb{I}$ denote the set of physical implementations Willie considers possible, where each element specifies both a target instrument and a physical realization of that instrument. A sufficient covertness condition is then $
\sup_{\tilde{\mathcal{I}}\in\mathbb{I}} \frac{1}{4}\left\|S_A^{(1)}(\tilde{\mathcal{I}})-S_A^{(0)}\right\|_{\diamond r}\le \delta.$
We defer exploration of this scenario to future work.
\end{remark}

\begin{remark}[Willie's Capabilities and Special Cases]
If Willie's interaction with the QPU is represented by a single input-output channel, corresponding to $r=1$, then Criterion~\ref{crit:covert-strategy-computation} is instead defined by the diamond-norm distance: $ \tfrac{1}{4}
\|\mathcal{N}^{(1)}(\tilde{\mathcal{I}}) - \mathcal{N}^{(0)}\|_\diamond \le \delta$, where $\mathcal{N}^{(1)}(\tilde{\mathcal{I}})$ and $\mathcal{N}^{(0)}$ are the effective channels he observes on his QCUs and reference system when Alice is computing or idle, respectively.

If Willie has no reference system and no adaptive access, but may choose the initial state $\hat{\sigma}_{W^s}$ of his QCUs, then $\delta$-covertness reduces to $\sup_{\hat{\sigma}_{W^s}}\tfrac{1}{4}\|\hat{\rho}^{(0)}_{W^s} - \hat{\rho}^{(1)}_{W^s}\|_1 \le \delta$.
\end{remark}

\subsection{Unreliable Computation and Multi-Shot Instruments}
\label{sec:unreliable-computation}

The covertness criterion established in the previous section applies to a single implementation of Alice's quantum instrument. However, lacking fault tolerance, current Noisy Intermediate-Scale Quantum (NISQ) devices typically require Alice to execute $k$ independent shots of her circuit and estimate the expectation values. The quantum-instrument formalism extends to this setting via reset operations interleaved between her per-shot implementation.

We can establish the worst-case bounds by allowing Willie's capabilities to exceed those of Alice. He has a perfect, indefinitely-coherent quantum memory of arbitrary size that persists across Alice's $k$-shot quantum instrument. He extends his strategy across all $k$ shots, denoted $S_W^k$. He maintains coherence between shots, coupling the unmeasured output of one shot with the next and defers a joint measurement over all systems he controls until after the final shot. This scenario is treated the same as that described in Section~\ref{sec:covertness}, where the round structure is extended to include Willie's actions across all $k$ shots.

If Willie cannot maintain a coherent memory across the $k$ shots, then $\hat{\rho}^{(b)}_{k}(S_W^k) = (\hat{\rho}^{(b)}(S_W))^{\otimes k}$, where $S_W$ is a single-shot strategy applied identically to each shot and $b\in\{0,1\}$. 
Quantum relative entropy (QRE), given by $D(\hat{\rho}\|\hat{\sigma})=\tr(\hat{\rho}(\log\hat{\rho}-\log\hat{\sigma}))$ for arbitrary quantum states $\hat{\rho}$ and $\hat{\sigma}$, is additive for product states.  Utilizing the quantum Pinsker's inequality \cite[Thm.~11.9.1]{wilde16quantumit2ed}, we have:
\begin{align}
&\sup_{S_W}\frac{1}{4}\|(\hat{\rho}^{(0)}_{W^sM}(S_W))^{\otimes k} -
(\hat{\rho}^{(1)}_{W^sM}(S_W))^{\otimes k}\|_1  \notag \\
&\phantom{\sup_{S_W}} \leq \sqrt{\frac{k}{8}\sup_{S_W}
D(\hat{\rho}^{(1)}_{W^sM}(S_W) \|
\hat{\rho}^{(0)}_{W^sM}(S_W))}, \label{eq:pinskers}
\end{align}
where the factor of $k$ on the right-hand side is due to additivity. Now, let $\delta_{\text{QRE}}$ be a bound on the QRE such that $\sup_{S_W} D(\hat{\rho}^{(1)}_{W^sM}(S_W) \|\hat{\rho}^{(0)}_{W^sM}(S_W)) \le\delta_{\text{QRE}}$. Setting $\delta_{\text{QRE}} = 8\delta^2$ ensures that any per-shot quantum instrument satisfying the QRE bound also satisfies Criterion~\ref{crit:covert-strategy-computation} for a single shot. For $k$ independent shots, the same per-shot bound yields
\begin{align}
\sup_{S_W}\frac{1}{4}\|(\hat{\rho}^{(1)}_{W^sM}(S_W))^{\otimes k} -(\hat{\rho}^{(0)}_{W^sM}(S_W))^{\otimes k}\|_1 \le \delta\sqrt{k}.
\end{align}
Thus, this gives $\delta\sqrt{k}$-covertness over $k$ shots.

In this case, we have granted Willie full knowledge of Alice's quantum instrument, and under the assumption of stable environment noise and no inter-shot memory effects, the per-shot states are i.i.d.~under each hypothesis. The quantum Stein lemma~\cite{ogawaStrongConverseSteins2000a} then characterizes the asymptotic decay of Willie's optimal missed-detection probability:
\begin{align}
P_{\rm MD}^{\star}(\alpha) = \exp\left[-k \sup_{S_W}D\left(\hat{\rho}^{(1)}(S_W) \middle\| \hat{\rho}^{(0)}(S_W) \right)+ o(k)\right],
\end{align}
where $P_{\rm MD}^{\star}(\alpha)$ denotes the minimum missed-detection probability over all tests satisfying $P_{\rm FA}\le\alpha$. Additionally, the terms hidden in $o(k)$ depend on $\alpha$ but vanish for $k\to\infty$.

\begin{remark}
The multi-shot setting also admits additional covertness methods and refinements for Alice. Here, Willie is attempting to detect if Alice is computing at any time over the $k$ shots. For a fixed shot budget, she may interleave idle shots trading expectation-value precision for a reduced per-shot detection signal, potentially altering Willie's optimal strategy. Bounds for QRE in covert signaling \cite{tahmasbi21signalingcovert, bullockFundamentalLimitsCovert2025, zlotnick25eacqcovert, kimelfeldCovertEntanglementGeneration2026} and perturbation theory \cite{grace2022perturbation} can be adapted to analyze this setting. Furthermore, in the composite-hypothesis setting of Remark~\ref{rem:unknown-i}, Willie does not know $\tilde{\mathcal{I}}$, so the alternate-hypothesis state $\hat{\rho}^{(1)}_{k}(S_W^{k})$ is a family of states rather than a single state. The generalized quantum Stein lemma is then required. We defer both explorations to future work.
\end{remark}

The capabilities granted to Willie in this section, including coherent memory, arbitrary operations per round, and coherence preserved across shots, exceed what NISQ devices support. Indeed, these devices not only limit Alice but also Willie.
We explore the impact of these limitations next.

\section{Square-Root Law for Covert Quantum Computing on Superconducting Architectures}
\label{sec:scaling-bound}

We specialize the formalism of the previous section to NISQ superconducting architectures. On these devices, Alice and Willie's QCUs are qubits arranged on planar lattices, gate fidelities and coherence times bound circuit strategy and depth, and coupling to an external memory is currently infeasible, but has been proposed as an architecture to scale beyond a few thousand qubits~\cite{bravyiFutureQuantumComputing2022, yoderTourGrossModular2025}. The quantum information passed to Willie at each round arises from crosstalk between Alice and Willie's qubits. Willie uses spectator qubits to sense these effects, which we describe next. We then characterize the dominant crosstalk mechanisms, and finally, we use the discrete isoperimetric inequalities to show that Alice can remain covert by employing only $\mathcal{O}(\sqrt{n})$ additional qubits.

\subsection{Spectator Qubits}

\emph{Spectator qubits} are ancillary qubits observed during a computation or idling, rather than used to encode data. These qubits are typically used to sense temporal drifts in controls \cite{majumderRealtimeCalibrationSpectator2020}, understand and correct miscalibrations \cite{guptaIntegrationSpectatorQubits2020}, and characterize coherent crosstalk between data qubits \cite{fangCrosstalkSuppressionIndividually2022}. In a typical quantum computing setting, this information can enable feedback-based error mitigation, either between experiments or computations, by updating calibration parameters, or via mid-circuit measurement and feedforward techniques.

In the covert quantum computing setting, Willie re-purposes spectator qubits to detect noise that deviates from the system's idle baseline. This allows him to infer whether other qubits on the QPU are undergoing active gate operations. Thus, Willie's detection of deviations in noise parameters allows detection of computation. The dominant source of such noise in superconducting architectures is crosstalk arising from two-qubit gate operations, which we characterize next.

\subsection{Crosstalk Noise}
\label{sec:noise-models}
Understanding the physical mechanisms that generate crosstalk is essential for characterizing Willie's detection problem. In superconducting qubit architectures, transmons are the predominant qubit technology used by IBM, Google, Rigetti, and IQM~\cite{kjaergaardSuperconductingQubitsCurrent2020,acharyaQuantumErrorCorrection2025,lanesFrameworkQuantumAdvantage2025,cauneDemonstratingRealtimeLowlatency2024, castelvecchiIBMReleasesFirstever2023}. These act as weakly anharmonic oscillators, such that while computation is based on two-level qubits, leakage into higher energy levels can occur. Single-qubit gates are driven by microwave pulses resonant with the $\ket{0} \leftrightarrow \ket{1}$ transition, and therefore frequency collisions between driven transitions and nearby qubit or coupler transitions must be minimized to suppress leakage and crosstalk.

Two-qubit gates are a dominant source of correlated crosstalk: they are slower than single-qubit gates and require activation of an interaction, such as through flux modulation, which increases susceptibility of nearby qubits to errors and crosstalk as well. One prominent two-qubit gate implementation in current superconducting hardware is the flux-tunable coupler architecture. Because our experiments in Section~\ref{sec:experiments} use IBM Heron and IQM Crystal-based processors, we focus on flux-tunable coupler implementations of the CZ gate. While we leave the full details of the physics underpinning tunable couplers to other works \cite{yanTunableCouplingScheme2018, sungRealizationHighFidelityCZ2021, marxerLongDistanceTransmonCoupler2023}, here we provide a brief description of their mechanisms and the crosstalk noise that motivates the nearest-neighbor model.

The interaction strength between the qubits is determined by both a direct coupling with strength $g_{12}$ and an indirect coupling via the coupler with strengths $g_{1c}$ and $g_{2c}$. The leading-order effective qubit–qubit coupling is:
\begin{align}
    \tilde{g} = g_{12} + \frac{g_{1c}g_{2c}}{2}\left(\frac{1}{\Delta_{1c}} + \frac{1}{\Delta_{2c}}\right), \label{eq:tunable-coupler-coupling-strength}
\end{align}
where $\Delta_{ic} = \omega_i - \omega_c$ is the detuning between qubit $i$ and the coupler. Tuning the coupler frequency $\omega_c$ via an external magnetic flux bias generates destructive interference between the two couplings, yielding $\tilde{g} \approx 0$, effectively decoupling the qubit pair at its idling point. The CZ gate is then applied by pulsing the coupler away from the idle point, activating an effective ZZ interaction between the two gate qubits. Up to single-qubit Z phases and a global phase, the interaction causes the $\ket{11}$ component to acquire an additional relative phase. Pulsing to generate a $\pi$ phase induces the CZ.

Even at the idling point, residual ZZ coupling may exist, not only between the intended gate pair, but also between gate qubits and nearest-neighbor spectators. During a two-qubit gate pulse, spectators can experience frequency shifts proportional to their ZZ coupling to the gate qubits. These shifts cause spectators to accumulate phase errors, providing Willie with a detectable coherent signature of Alice's gate operation. This nearest-neighbor mechanism is the primary topology-dependent crosstalk channel modeled in our square-root law for covert quantum computing. Non-nearest-neighbor couplings may also be present and, in theory, used by Willie for detection. However, in tunable-coupler designs, they are expected to be significantly weaker than nearest-neighbor couplings, with next-nearest-neighbor couplings reported at least two orders of magnitude smaller in IQM's architecture~\cite[Sec.~II]{marxerLongDistanceTransmonCoupler2023}. We therefore emphasize that this is not an exhaustive model of crosstalk or noise within these systems. Thus, our experiments test not only the predicted nearest-neighbor crosstalk detection, but also the validity of the locality assumption itself.

\subsection{Superconducting Qubit Topologies}
\begin{figure}[htbp]
    \centering
    \begin{minipage}[b]{0.48\linewidth}
        \centering
        \includegraphics[width=\linewidth,height=3.2cm,keepaspectratio]{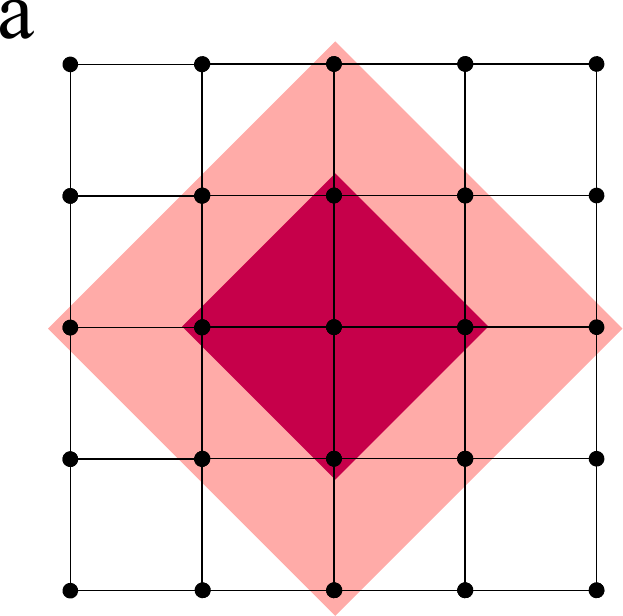}
    \end{minipage}\hfill
    \begin{minipage}[b]{0.48\linewidth}
        \centering
        \includegraphics[width=\linewidth,height=3.2cm,keepaspectratio]{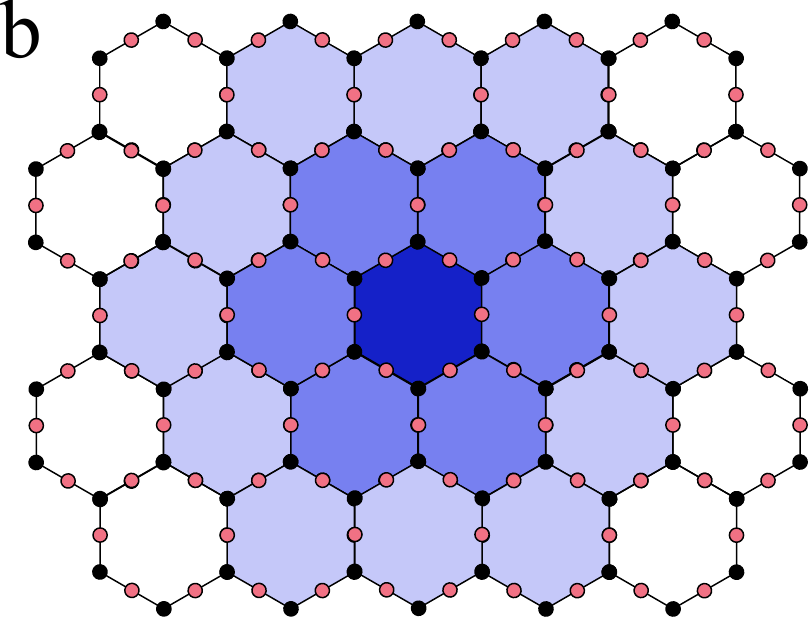}
    \end{minipage}
    \caption{\textbf{(a)}~The square lattice is shown with black vertices and edges. The shape that optimizes the vertex-isoperimetric inequality is a diamond: the dark red diamond encloses $5$ vertices, and the light red diamond encloses $13$. \textbf{(b)}~The hexagonal lattice is shown with black vertices and edges; the heavy-hex lattice is obtained by adding the red vertices. The shaded regions show the disk-construction vertex subsets of the hexagonal lattice for radii $r=0$, $1$, and $2$ (dark, medium, and light shading, respectively), where each layer subsumes all interior layers. In the heavy-hex disk construction, each outgoing edge and its associated red vertex are assigned to the preceding layer. }
    \label{fig:square-hex-topologies}
\end{figure}

Under the noise models described above, crosstalk affects the nearest-neighbor qubits only. It is typically several orders of magnitude weaker for non-nearest neighbors and is virtually undetectable. Thus, at least in theory, Willie can obtain information on whether Alice is computing \emph{only} from the boundary of her circuit. Therefore, Alice can prevent Willie from using any strategy to determine whether she is computing by leaving the border qubits idle. This results in maximum detection error probability $P^{\text{(e)}}_W=\frac{1}{2}$ for Willie, at the cost of idle-qubit overhead to Alice. We now study how superconducting quantum chip layouts affect this cost.

The predominant topologies for superconducting qubits are the square \cite{abdurakhimovTechnologyPerformanceBenchmarks2024} and the heavy-hex two-dimensional (2D) lattices shown in Fig.~\ref{fig:square-hex-topologies}. The heavy-hex lattice was pioneered by IBM to reduce frequency collisions and errors in a circuit by a slight decrease in connectivity~\cite{IBMQuantumHeavy}.
\emph{Isoperimetric inequalities} provide a tight lower bound on the boundary size for a given surface area.
Next, we explore their discrete versions and show that, for current 2D quantum-computer topologies, we need only $\mathcal{O}(\sqrt{n})$ qubits to isolate Alice's $n$-computational-qubit circuit from Willie.

\subsection{Isoperimetric Analysis}
The isoperimetric inequality involves closed curves in a two-dimensional plane, and states $L^2 \le 4\pi A$, where $L$ is the length of the closed curve and $A$ is the area it encloses. The isoperimetric inequality in the plane states that a circle maximizes the area-to-perimeter ratio. Here, we require a discrete version of the isoperimetric inequality for sets of vertices embedded on a graph or lattice.

We define a graph $G=(V(G), E(G))$ as the pair consisting of vertex set $V(G)$ and edge set $E(G)$. For a subset $S\subseteq V(G)$, the \emph{vertex boundary} $N(S)$ is the set of vertices not in $S$ that have an edge which connects it to a vertex in $S$: $N(S) \equiv \{u \in V(G) \setminus S \mid \{u, v\} \in E(G) \wedge v \in S\}$. The \emph{edge boundary} $\mathcal{E}(S)$ is the set of edges that connect a vertex in $S$ to a vertex not in $S$: $\mathcal{E}(S) \equiv \{\{u, v\} \in E(G) \mid u \in S \wedge v \in V(G) \setminus S\}$. The cardinality $|S|$ of the vertex set $S$ plays the role of area in the continuous case.

\subsubsection{Square Grid}
Denote $G^{\rm S}_\infty$ as the infinite square grid, then for any finite subset $S^{\rm S} \subset V(G^{\rm S}_\infty)$, the isoperimetric inequalities are ~\cite{harperOptimalAssignmentsNumbers1964, bollobasEdgeisoperimetricInequalitiesGrid1991a}:
\begin{align}
|\mathcal{E}(S^{\rm S})|&\ge 4\sqrt{|S^{\rm S}|}, \\
|N(S^{\rm S})|&\ge 2\sqrt{2|S^{\rm S}|}, \label{eq:square-isoperimetric}
\end{align}
with equality achieved by diamond-shaped (Fig.~\ref{fig:square-hex-topologies}(a)) and square-shaped finite graphs, for the vertex and edge boundary, respectively.

\subsubsection{Hexagonal Grid} Denote $G^{\rm H}_\infty$ as the infinite hexagonal grid. Then for any finite subset $S^{\rm H} \subset V(G^{\rm H}_\infty)$~\cite{grussienIsoperimetricInequalitiesHexagonal2012}:
\begin{align}
    |N(S^{\rm H})| &\geq \sqrt{6|S^{\rm H}|}, \label{eq:hex-isoperimetric-vertex}\\
    |\mathcal{E}(S^{\rm H})| &\geq \sqrt{6|S^{\rm H}|}. \label{eq:hex-isoperimetric-edge}
\end{align}
The hexagonal disk $G^{\rm H}_r$ with radius $r$ is a finite graph given by concentric layers of regular hexagons with a single hexagon having radius $r=0$. The shaded regions excluding the red vertices in Fig.~\ref{fig:square-hex-topologies}(b) demonstrate $G^{\rm H}_r$ for $r=0,1,2$ in dark, medium, and light blue. The hexagonal disk construction achieves equality in both \eqref{eq:hex-isoperimetric-vertex} and \eqref{eq:hex-isoperimetric-edge}.

\subsubsection{Heavy-Hex Grid}
We derive the isoperimetric inequalities for the heavy hexagonal grid denoted $G^{\rm HH}_\infty$, by first demonstrating a general result that relates the cardinality of the vertex and edge boundaries for a connected planar graph and its related subdivision graph generated by inserting a vertex on each edge.

\begin{lemma}\label{lem:generalize-heavy-isoperimetric}
Let $G^\prime$ be an arbitrary graph, and let $G$ be obtained from $G^\prime$ by subdividing every edge exactly once, inserting a degree-two vertex $v_{uw}$ for each edge $\{u,w\} \in E(G^\prime)$. Let $S_2 = \{v_{uw} : \{u,w\} \in E(G^\prime)\}$ be the set of inserted vertices such that $V(G) = V(G^\prime) \cup S_2$. Then for any $S^\prime \subseteq V(G^\prime)$ and any $S \subseteq V(G)$ with $S \cap V(G^\prime) = S^\prime$,
\begin{align}
    |N(S)| &\geq |N(S^\prime)|, \label{eq:subdivision-isoperimetric-vertex}\\
    |\mathcal{E}(S)| &\geq |\mathcal{E}(S^\prime)|, \label{eq:subdivision-isoperimetric-edge}
\end{align}
where $N(\cdot)$ and $\mathcal{E}(\cdot)$ are the vertex and edge boundaries in $G$ and $G^\prime$ respectively. Equality holds in \eqref{eq:subdivision-isoperimetric-vertex} when $S = S^\prime \cup \{v_{uw} \in S_2 : u \in S^\prime \text{ or } w \in S^\prime\}$, and in \eqref{eq:subdivision-isoperimetric-edge} when $S = S^\prime \cup \{v_{uw} \in S_2 : u \in S^\prime \text{ and } w \in S^\prime\}$.
\end{lemma}
The proof is in Appendix~\ref{app:subdivision-lemma-proof}.

Combining Lemma~\ref{lem:generalize-heavy-isoperimetric} with the hexagonal-grid isoperimetric inequality \eqref{eq:hex-isoperimetric-vertex} yields:
\begin{lemma}\label{lem:heavy-hex-vertex-isoperimetric}
For any finite subset $S^{\rm HH} \subset V(G^{\rm HH}_\infty)$,
\begin{align}
    |N(S^{\text{HH}})| \geq \frac{-9 + \sqrt{81 + 60|S^{\text{HH}}|}}{5}. \label{eq:heavy-hex-vertex-isoperimetric}
\end{align}
\end{lemma}
\begin{proof}
Let $S^{\rm H} = S^{\rm HH} \cap V(G^{\rm H}_\infty)$ denote the hexagonal-grid vertices in $S^{\rm HH}$. Lemma~\ref{lem:generalize-heavy-isoperimetric} together with \eqref{eq:hex-isoperimetric-vertex} yields
\begin{align}
    |N(S^{\rm HH})| \geq |N(S^{\rm H})| \geq \sqrt{6|S^{\rm H}|}. \label{eq:hh-proof-chain}
\end{align}
Since $G^{\rm H}_\infty$ is $3$-regular, each subdivision vertex in $S^{\rm HH}$ corresponds to an edge of $G^{\rm H}_\infty$ whose vertices both lie in $S^{\rm H} \cup N(S^{\rm HH})$, thus,
\begin{align}
    |S^{\rm HH}| &= |S^{\rm H}| + |S^{\rm HH} \cap S_2| \\
     &\leq |S^{\rm H}| + \frac{3\left(|S^{\rm H}| + |N(S^{\rm HH})|\right)}{2} \\
     &= \tfrac{5}{2}|S^{\rm H}| + \tfrac{3}{2}|N(S^{\rm HH})|.
\end{align}
Substituting the solution for $|S^{\rm H}|$ in \eqref{eq:hh-proof-chain} and taking the positive root yields \eqref{eq:heavy-hex-vertex-isoperimetric}.
\end{proof}

\begin{remark}
    Lemma~\ref{lem:generalize-heavy-isoperimetric} is not specific to the heavy-hex grid. In particular, it yields isoperimetric bounds on the sizes of vertex and edge boundaries for any once-subdivided graph $G$ to those on its non-subdivided base, $G^\prime$. 
    For example, applying the lemma to the square grid yields an analogous bound for the heavy-square lattice, which is constructed by inserting a vertex on every edge of the square lattice.
\end{remark}

\subsection{Square-Root Cost Bound for Covert Quantum Computing}
The isoperimetric inequalities in~\eqref{eq:square-isoperimetric}, \eqref{eq:hex-isoperimetric-vertex} and \eqref{eq:heavy-hex-vertex-isoperimetric} establish that any finite vertex set $S$ on an infinite square, hexagonal, or heavy-hex lattice satisfies $|N(S)| \geq c\sqrt{|S|}$ for a lattice-dependent constant $c$. The shapes of optimal sets achieving equality are diamond for the square grid~\cite{harperOptimalAssignmentsNumbers1964} and hexagonal disk for the hexagonal grid. For the heavy-hex grid, the optimal shape is an open problem, but we conjecture it is the hexagonal-disk construction with all subdivision vertices included. From Lemma~\ref{lem:heavy-hex-vertex-isoperimetric}, the leading-order scaling is $\sqrt{12|S|/5}$. These results establish that compact placements achieving $|N(S)| = \mathcal{O}(\sqrt{|S|})$ exist. Combined with the nearest-neighbor noise models from Section~\ref{sec:noise-models}, this yields a square-root bound on the number of additional qubits Alice must use to remain covert.

\begin{corollary}\label{cor:square-root-bound-for-covert-computing}
When unintended crosstalk noise affects only nearest-neighbor qubits, any circuit on $n$ qubits embedded in a sufficiently large $N$-qubit QPU with a square, hexagonal, or heavy-hex topology admits a placement that exposes at most $\mathcal{O}(\sqrt{n})$ spectator qubits to detectable crosstalk. Therefore, at most $\mathcal{O}(\sqrt{n})$ additional qubits must be included in the circuit as idling qubits to ensure covertness.
\end{corollary}
\begin{proof}
For the square grid, let $k = \lceil\sqrt{n}\rceil$ and place the circuit on a $k \times k$ subgrid, selecting $n$ of the $k^2$ qubits for computation and idling the remaining $k^2 - n \leq 2\sqrt{n}$. The vertex boundary of the $k \times k$ block satisfies $|N(S)| \leq 4k = 4\lceil\sqrt{n}\rceil$. Hence, the total number of additional qubits Alice must include as idling qubits is at most $(k^2 - n) + |N(S)| \leq 2\sqrt{n} + 4\lceil\sqrt{n}\rceil = \mathcal{O}(\sqrt{n})$. The analogous constructions on hexagonal and heavy-hex grids use hexagonal disks of the smallest radius $r$ with vertex cardinality larger than $n$ and $\mathcal{O}(\sqrt{n})$ scaling given by Lemma~\ref{lem:heavy-hex-vertex-isoperimetric}.%
\end{proof}

This corollary states that, provided Alice's computation affects only nearest-neighbor qubits, she may remain covert against the most capable Willie described in Section~\ref{sec:covert-quantum-computing-theory}, and not just the one limited to NISQ hardware.
This is because her computational qubits are entirely decoupled from Willie under the given noise model. Furthermore, we expect our approach to apply beyond the planar 2D topologies to future 3D quantum processors. However, we anticipate $\mathcal{O}(n^{2/3})$ rather than $\mathcal{O}(\sqrt{n})$ scaling, matching that of the geometric constraint on 3D-network transport capacity \cite{franceschettiCapacityWirelessNetworks2009a} and isoperimetric inequalities for the cubic lattice \cite{bollobasEdgeisoperimetricInequalitiesGrid1991a}.

We now turn to experimentally testing Corollary~\ref{cor:square-root-bound-for-covert-computing}.

\section{Experimental Validation}
\label{sec:experiments}

Corollary~\ref{cor:square-root-bound-for-covert-computing} states that Alice remains covert against an arbitrarily powerful Willie by mapping her $n$ computational qubits to a compact region of a lattice and idling an additional $\mathcal{O}(\sqrt{n})$ qubits along its vertex boundary, provided the crosstalk is confined to nearest neighbors. We experimentally test the scaling bound and the crosstalk locality premise on IQM's \emph{Emerald} and IBM's \emph{ibm\_fez} superconducting QPUs with square and heavy-hex topologies, respectively.  

We choose Willie's strategy to reflect the capabilities of current NISQ hardware. We perform Ramsey interferometry on each spectator qubit, since this protocol is well-suited to detecting frequency shifts induced by residual ZZ coupling from CZ gates \cite{krantzQuantumEngineersGuide2019, wagnerOptimizedNoiseSuppression2025}. Beyond sensitivity, Ramsey circuits minimize the impact of NISQ limitations: Willie only requires three single-qubit gates and a single measurement per spectator over the course of Alice's computation. We further limit Willie to no coherent quantum memory, and while strategies more sensitive to the noise model may exist, we defer their exploration to future work.

\subsection{Ramsey Protocol and Detection Strategy}
\label{sec:ramsey-protocol-and-detection-strategy}
Ramsey experiments exploit the fact that, under the Schr\"{o}dinger equation, the relative phase between $\ket{0}$ and $\ket{1}$ accumulates at a rate proportional to the qubit's transition frequency. The protocol prepares a spectator qubit in a superposition via a $\sqrt{X}$ gate, allows free evolution for a delay $\tau$, applies an $R_Z(2\pi f_{\rm osc}\tau)$ rotation where $f_{\rm osc}$ is a user-defined detuning, and applies a second $\sqrt{X}$ gate converting the relative phase into a measurable population difference. The generated circuit is on qubit $q_0$ in Fig.~\ref{fig:ramsey-circuits}. Sweeping $\tau$ produces oscillations in the measurement probability at the detuning frequency $\Delta = \omega_q - \omega_{\rm ref}$:
\begin{align}
    P(\ket{0}) = \frac{1}{2}\left(1 + e^{-\tau\Gamma}\cos(\Delta\tau)\right), \label{eq:ramsey-exp-equation}
\end{align}
where $\Gamma = 1/T_2^*$ is the dephasing rate and $T_2^*$ is the inhomogeneous dephasing time. Each experiment sweeps over $\tau$ at $39$ equally-spaced points from 0 to 11~$\mu$s with $f_{\rm osc} = 300$~kHz and 1024 shots per point. The detuning frequency $\Delta$ is then determined by fitting \eqref{eq:ramsey-exp-equation} using the measurement outcomes.

To detect crosstalk, we modify the baseline circuit by applying CZ gates to a neighboring qubit pair during the delay ($q_1$ and $q_2$ in Fig.~\ref{fig:ramsey-circuits}). The number of CZ gates is set by the ratio of the delay to the CZ gate time. Residual ZZ coupling between the gate qubits and the spectator causes the spectator to accumulate additional phase, appearing as a frequency offset in the Ramsey fit relative to the baseline run. The frequency shift for a given spectator is the difference between the detuning frequency during the baseline experiment and that during CZ application. An example of these results is shown in Fig.~\ref{fig:ramsey-experiment-1} for qubit~35 in Experiment~1 on \emph{Emerald}: blue circles and solid line are the baseline, red squares and dashed line show the crosstalk. The observed shift is $91.30$~kHz.

\begin{figure}
    \centering
    \includegraphics[width=\linewidth]{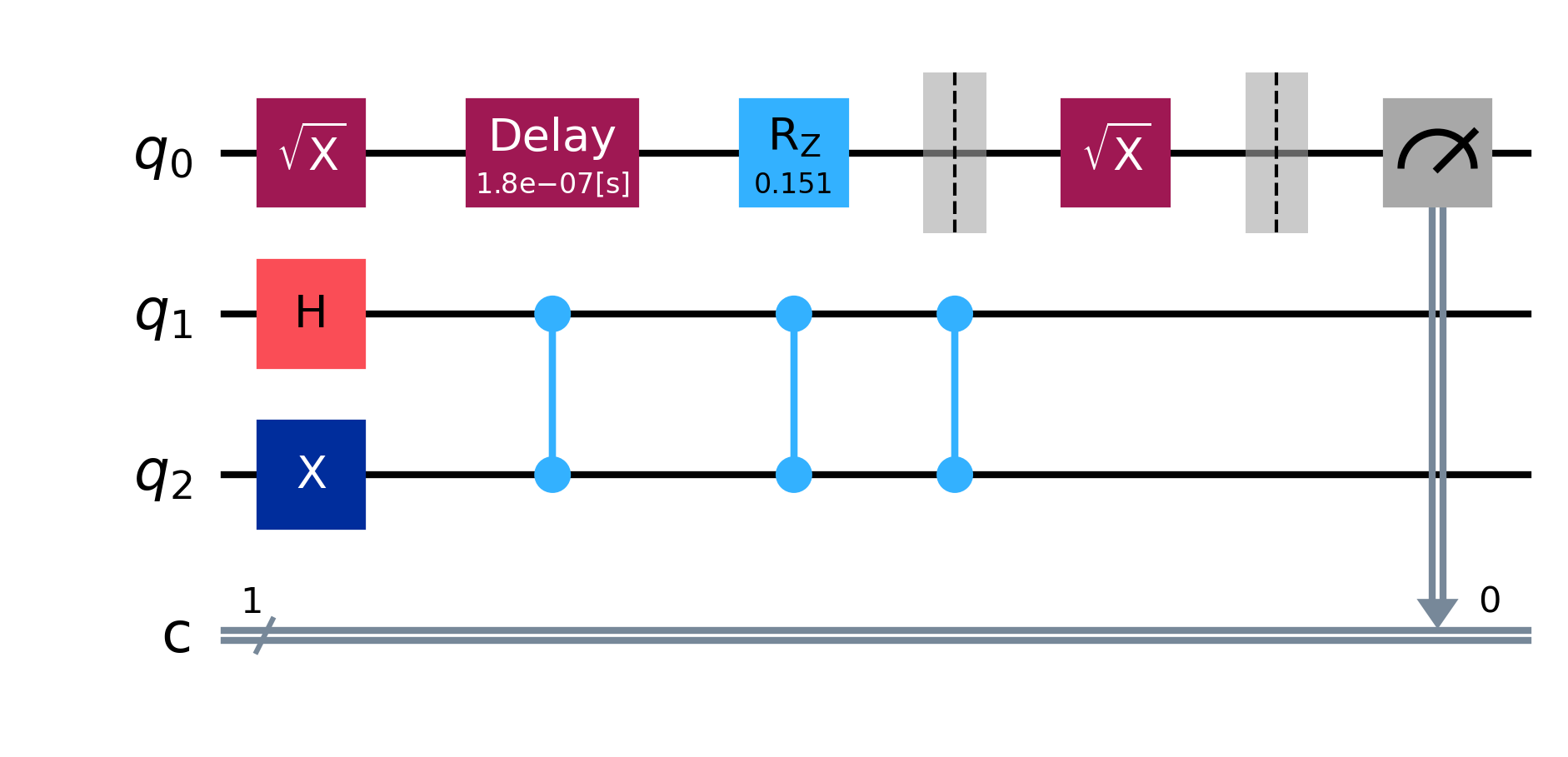}
    \caption{Experiment circuits. The gate sequence on qubit $q_0$ represents the Ramsey circuit performed on all spectator qubits. A $\sqrt{X}$ gate prepares a superposition, followed by a variable delay and $R_Z(2\pi f_{\rm osc}\tau)$ rotation. A second $\sqrt{X}$ gate is applied before measurement. Qubits $q_1$ and $q_2$ represent Alice's computation, applying repeated CZ gates during the Ramsey delay. Baseline results are gathered by letting qubits $q_1$ and $q_2$ idle.}
    \label{fig:ramsey-circuits}
\end{figure}

\begin{figure}
    \centering
    \includegraphics[width=\linewidth]{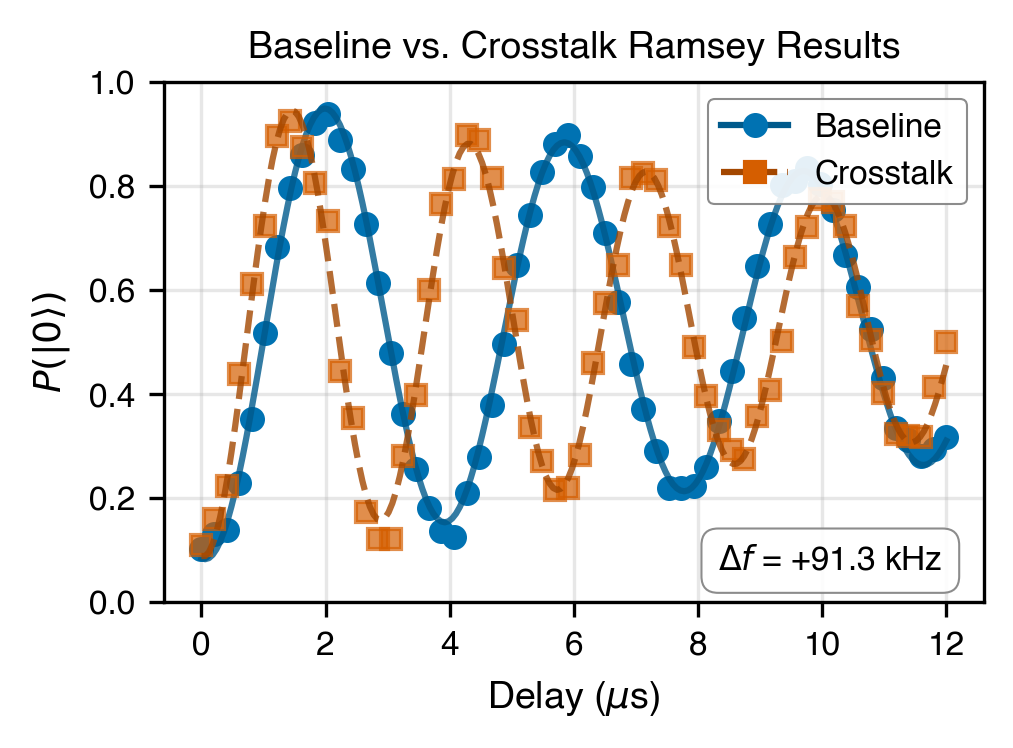}
    \caption[Ramsey experiment results for qubit 35 in \emph{Emerald} Experiment 1.]{Ramsey experiment results for qubit 35 in \emph{Emerald} Exp.~1. The horizontal axis corresponds to the delay in the Ramsey experiment circuit, while the vertical axis gives the probability of measuring $\ket{0}$, from \eqref{eq:ramsey-exp-equation}. A detectable frequency shift of $91.30$~kHz between the baseline and crosstalk results is measured.}
    \label{fig:ramsey-experiment-1}
\end{figure}

A qubit cannot simultaneously participate in more than one CZ gate. Therefore, for each experiment, Alice's CZ edges are divided into independent sets indicated by orange, blue, green, and maroon edges as seen in Fig.~\ref{fig:exp-topologies}. Each edge set is thus run independently, and detection results are aggregated across edge sets within an experiment.

A spectator qubit belonging to Willie is classified as a nearest neighbor (NN) if it is adjacent to any active CZ pair in an edge set, and non-NN otherwise. Therefore, while a qubit may be counted as a NN in one edge set, it may be possible to be counted as a non-NN in another.

Each experimental configuration is first run with two baseline circuits across all of Willie's qubits. A frequency shift $\Delta f$ is classified as detectable if $|\Delta f|$ exceeds twice the standard deviation of baseline-to-baseline frequency differences across all configurations on that device. This threshold is $12.85$~kHz for the \emph{Emerald} machine and $3.15$~kHz for \emph{ibm\_fez}.

We run experiments on two devices: IQM's 54-qubit \emph{Emerald}~\cite{abdurakhimovTechnologyPerformanceBenchmarks2024}, which uses a square lattice, and IBM's 156-qubit \emph{ibm\_fez}, which uses a heavy-hex lattice~\cite{IBMQuantumHeavy}. Both employ tunable-coupler architectures with CZ as the native two-qubit gate. Each experiment uses 1024 shots, $\tau$ sampled at 39 equally-spaced points from 0 to 11~$\mu$s, and $f_{\rm osc} = 300$~kHz. We construct and submit circuits to both backends using IBM Qiskit, and perform Ramsey fits using SciPy's \texttt{curve\_fit}.

A Ramsey fit on a single baseline measurement or single qubit, edge-set pair is rejected if \texttt{curve\_fit} fails to converge. For the \emph{Emerald} experiments, this corresponded to 3.9\% of fits, while only 2.7\% of fits failed to converge across all experiments on \emph{ibm\_fez}. Qubit 45 on \emph{Emerald} was excluded from results across all baseline and edge-set experiments, as it consistently fit as pure decoherence regardless of the surrounding gate activity.

\subsection{Experiment Configurations}
\label{sec:experiment-configs}
Table~\ref{tab:experiment-configs} describes the experiment configurations for each device. Figure~\ref{fig:exp-topologies} shows the physical mapping of the experiments to each QPU. Alice's CZ qubits are shown in red, and the edge sets where CZ gates are applied are indicated by like-colored edges; Willie's qubits are in black.

We present five experiments on \emph{Emerald} with expanding square grids of Alice's computation ($n = 2, 4, 9, 16, 25$). Square grids are used rather than the vertex-isoperimetrically optimal diamond shape because the diamond has a larger spatial footprint, e.g., for $n=25$, the diamond requires a $7\times 7$ bounding box versus $5\times 5$ for the square placement, which would reach the boundaries of the 54-qubit \emph{Emerald} layout with fewer data points. Nevertheless, both shapes yield $\mathcal{O}(\sqrt{n})$ vertex-boundary scaling.

For \emph{ibm\_fez}, we present four experiments with $n = 2, 12, 54, 108$. Experiments~2--4 correspond to heavy-hex disk constructions of radii $r = 1, 2, 3$, shown in Fig.~\ref{fig:square-hex-topologies}(b). We omit the outgoing edge and red vertex from each construction for compactness on the device, as we did with the square-over-diamond choice on \emph{Emerald}. Intermediate-$n$ runs ($n = 4, 6, 14, 21$) produce the same qualitative trends and are omitted for compactness.

\begin{table}[ht]
\centering
\caption{Experiment configurations on both devices.}
\label{tab:experiment-configs}
\begin{tabular}{cc|ccc|c}
\hline
\multirow{2}{*}{Device} & \multirow{2}{*}{Exp.} & \multicolumn{3}{c|}{Alice} & Willie \\
\cline{3-6}
& & \makecell{Num.\\qubits} & \makecell{Num.\\CZ edges} & \makecell{Edge sets\\(sizes)} & \makecell{Num.\\qubits} \\
\hline
\multirow{5}{*}{\emph{Emerald}}
 & 1 & 2   & 1   & 1               & 52 \\
 & 2 & 4   & 4   & 2, 2            & 50 \\
 & 3 & 9   & 12  & 4, 3, 4, 1      & 45 \\
 & 4 & 16  & 24  & 7, 6, 6, 5      & 38 \\
 & 5 & 25  & 40  & 10, 11, 9, 10   & 29 \\
\hline
\multirow{4}{*}{\emph{ibm\_fez}}
 & 1 & 2   & 1   & 1               & 147 \\
 & 2 & 12  & 12  & 6, 6            & 138 \\
 & 3 & 54  & 60  & 24, 22, 14      & 97 \\
 & 4 & 108 & 122 & 46, 46, 30      & 46 \\
\hline
\end{tabular}
\end{table}

\begin{figure*}[htbp]
    \centering
    \begin{subfigure}[b]{0.18\textwidth}
        \centering
        \includegraphics[width=\linewidth]{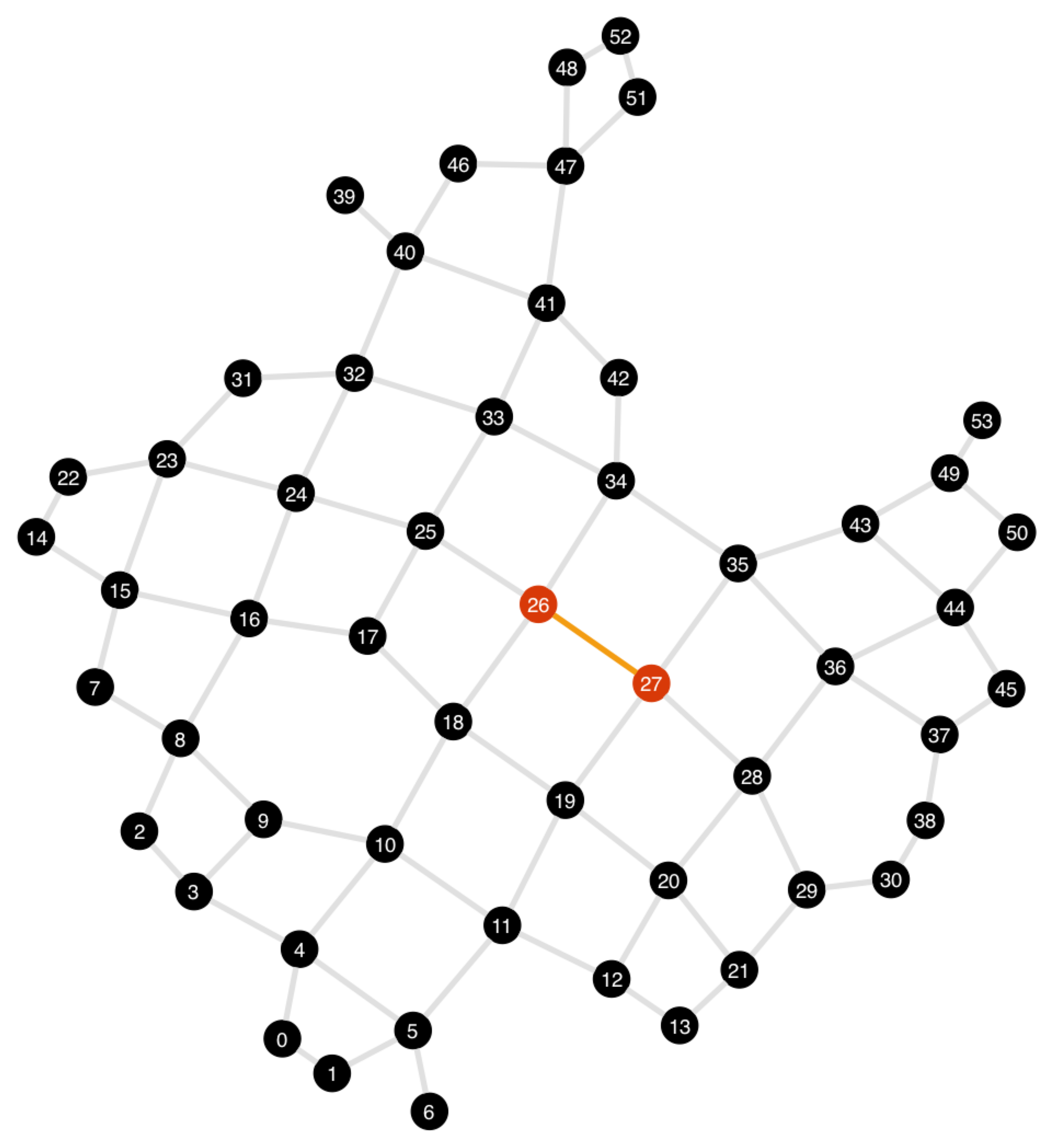}
        \caption{\emph{Emerald} Exp.~1.}
        \label{fig:exp-topology-iqm-1}
    \end{subfigure}\hfill
    \begin{subfigure}[b]{0.18\textwidth}
        \centering
        \includegraphics[width=\linewidth]{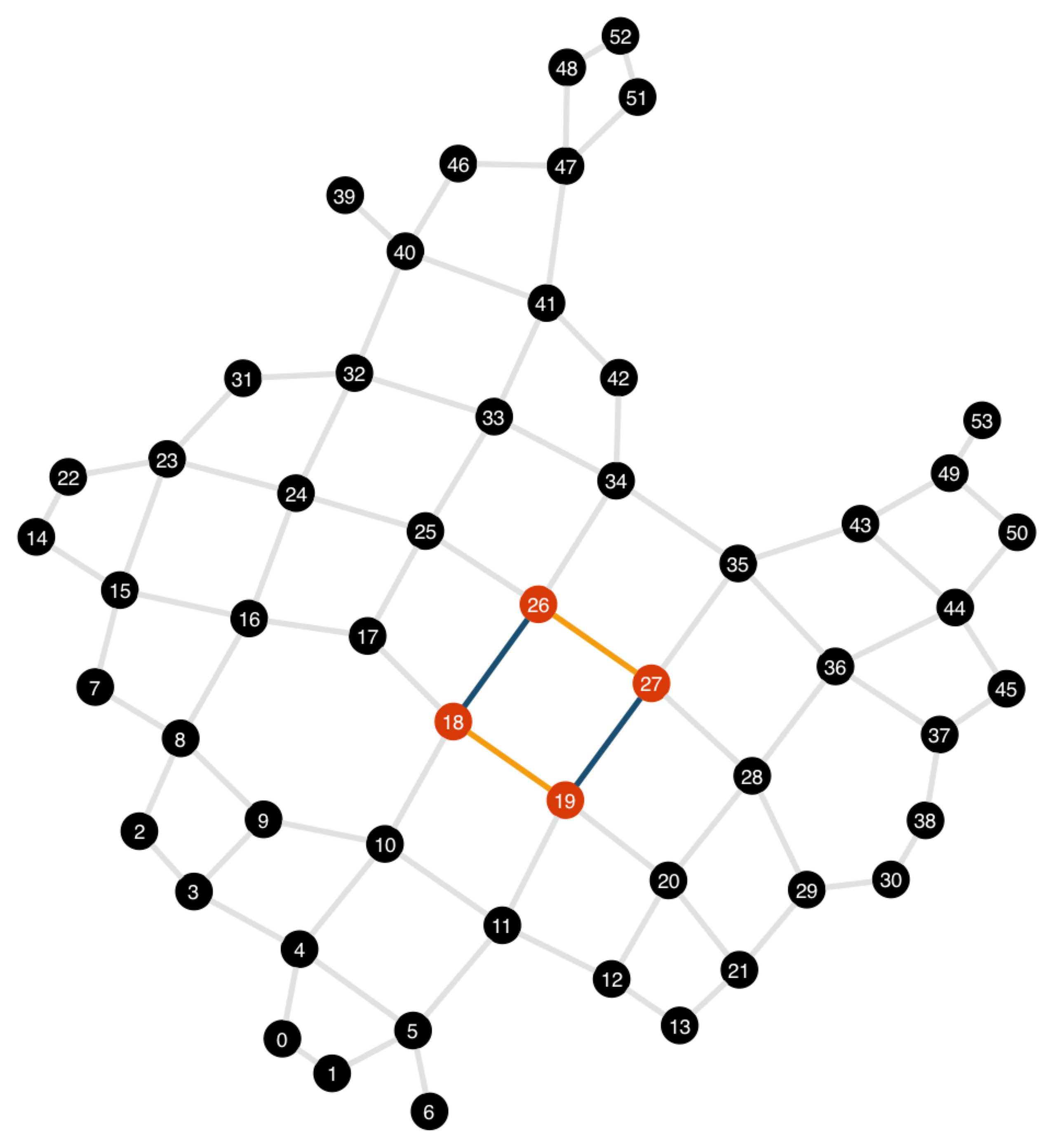}
        \caption{\emph{Emerald} Exp.~2.}
        \label{fig:exp-topology-iqm-2}
    \end{subfigure}\hfill
    \begin{subfigure}[b]{0.18\textwidth}
        \centering
        \includegraphics[width=\linewidth]{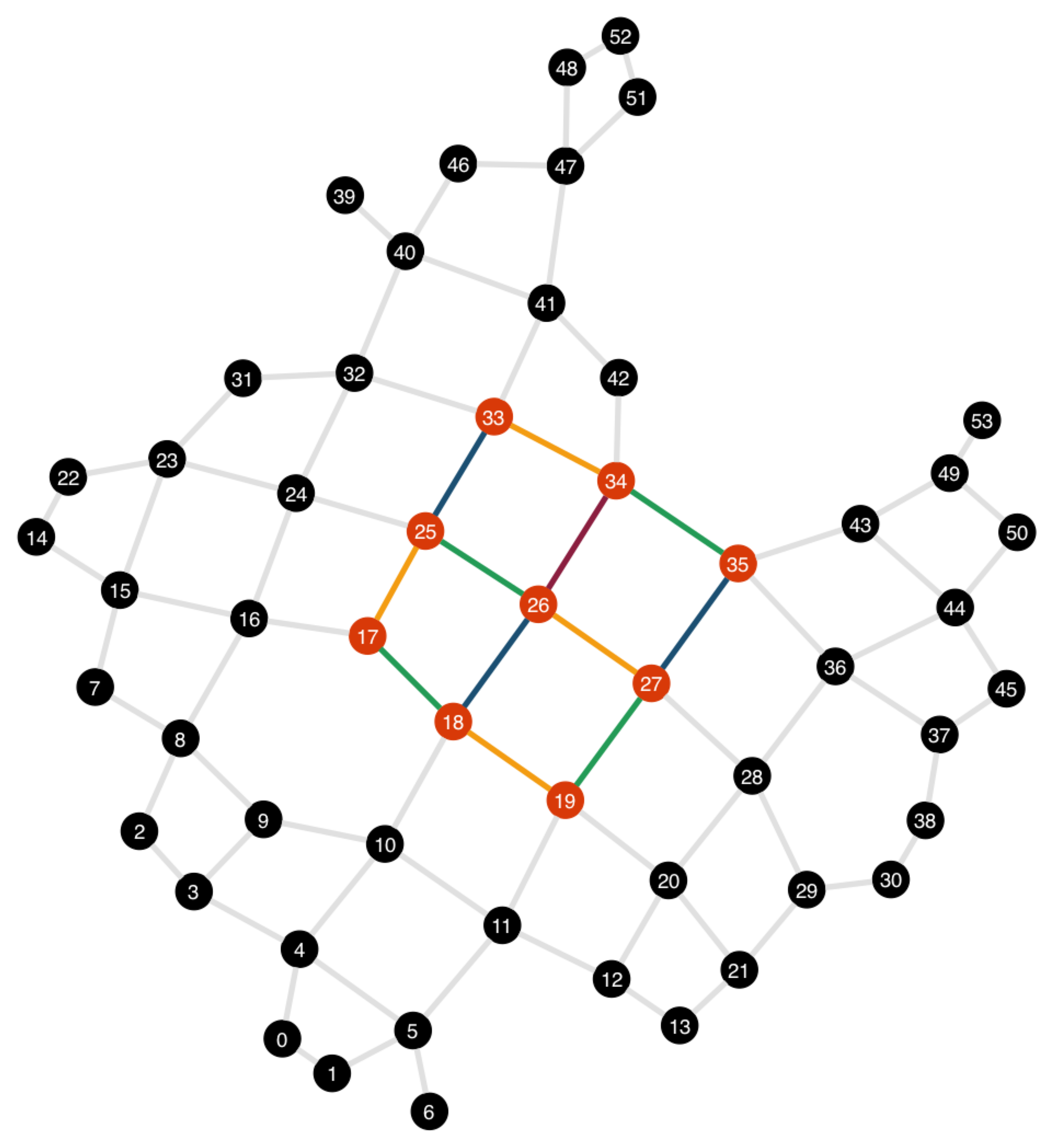}
        \caption{\emph{Emerald} Exp.~3.}
        \label{fig:exp-topology-iqm-3}
    \end{subfigure}\hfill
    \begin{subfigure}[b]{0.18\textwidth}
        \centering
        \includegraphics[width=\linewidth]{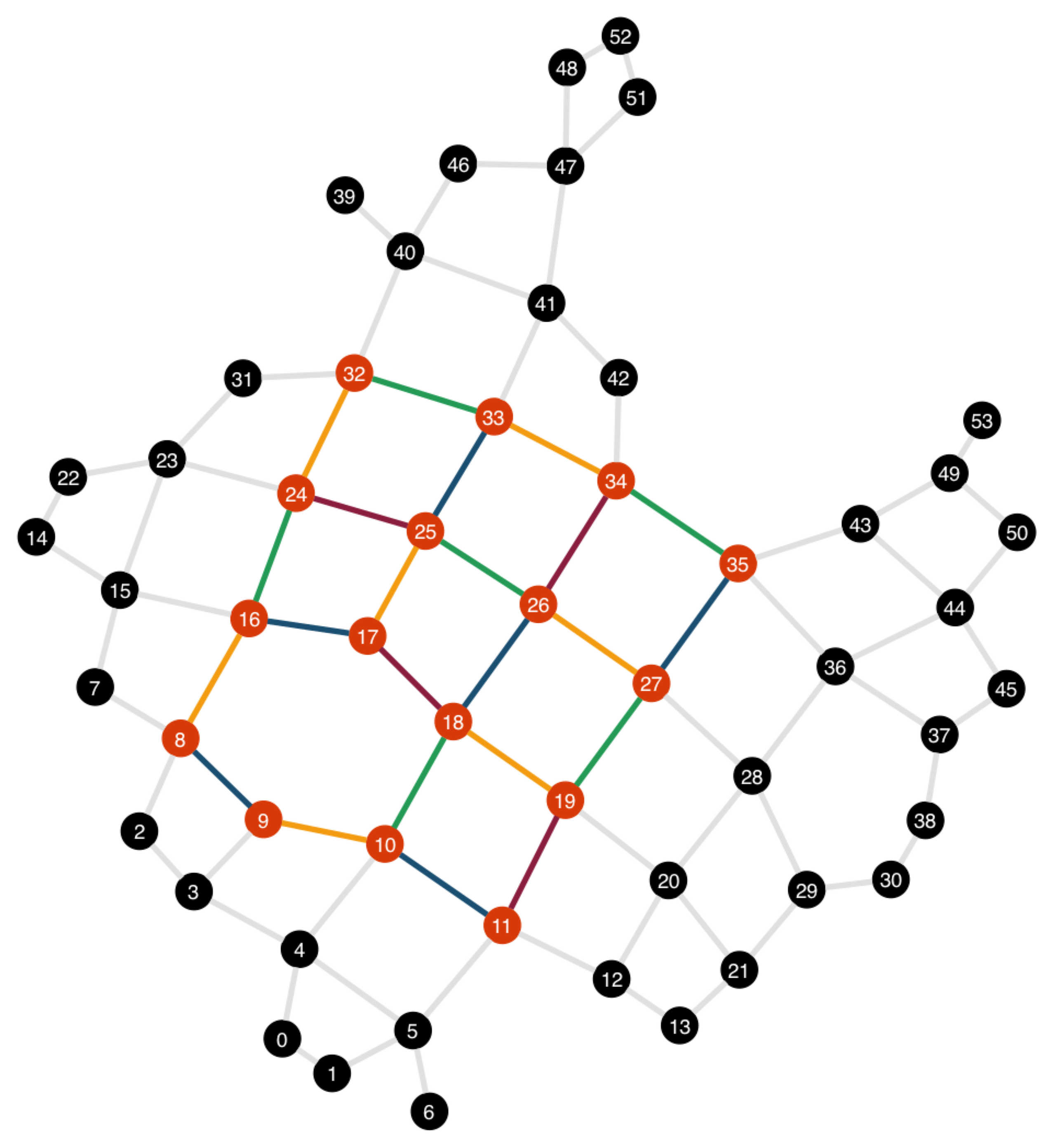}
        \caption{\emph{Emerald} Exp.~4.}
        \label{fig:exp-topology-iqm-4}
    \end{subfigure}\hfill
    \begin{subfigure}[b]{0.18\textwidth}
        \centering
        \includegraphics[width=\linewidth]{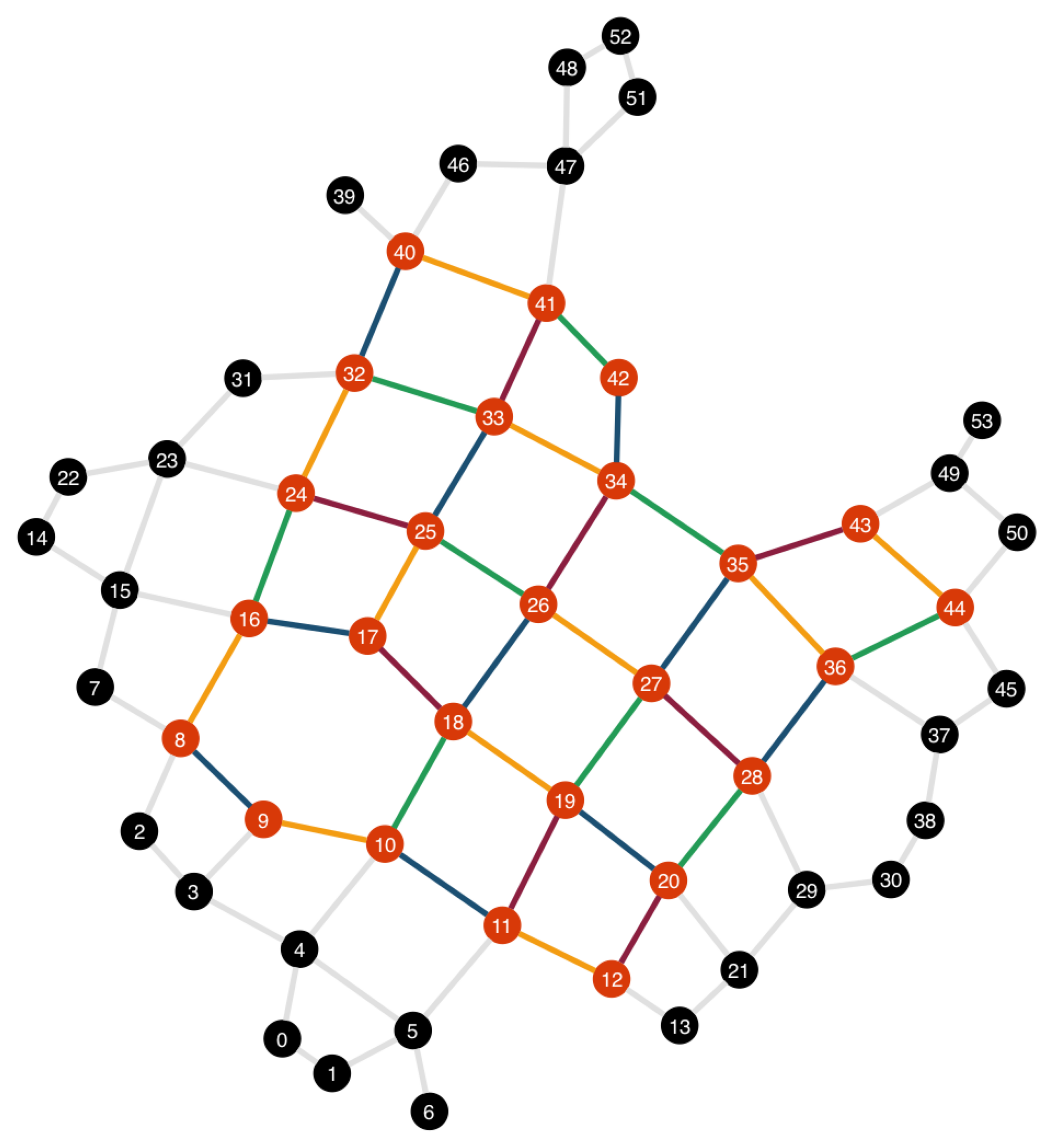}
        \caption{\emph{Emerald} Exp.~5.}
        \label{fig:exp-topology-iqm-5}
    \end{subfigure}

    \vspace{0.75em}

    \begin{subfigure}[b]{0.23\textwidth}
        \centering
        \includegraphics[width=\linewidth]{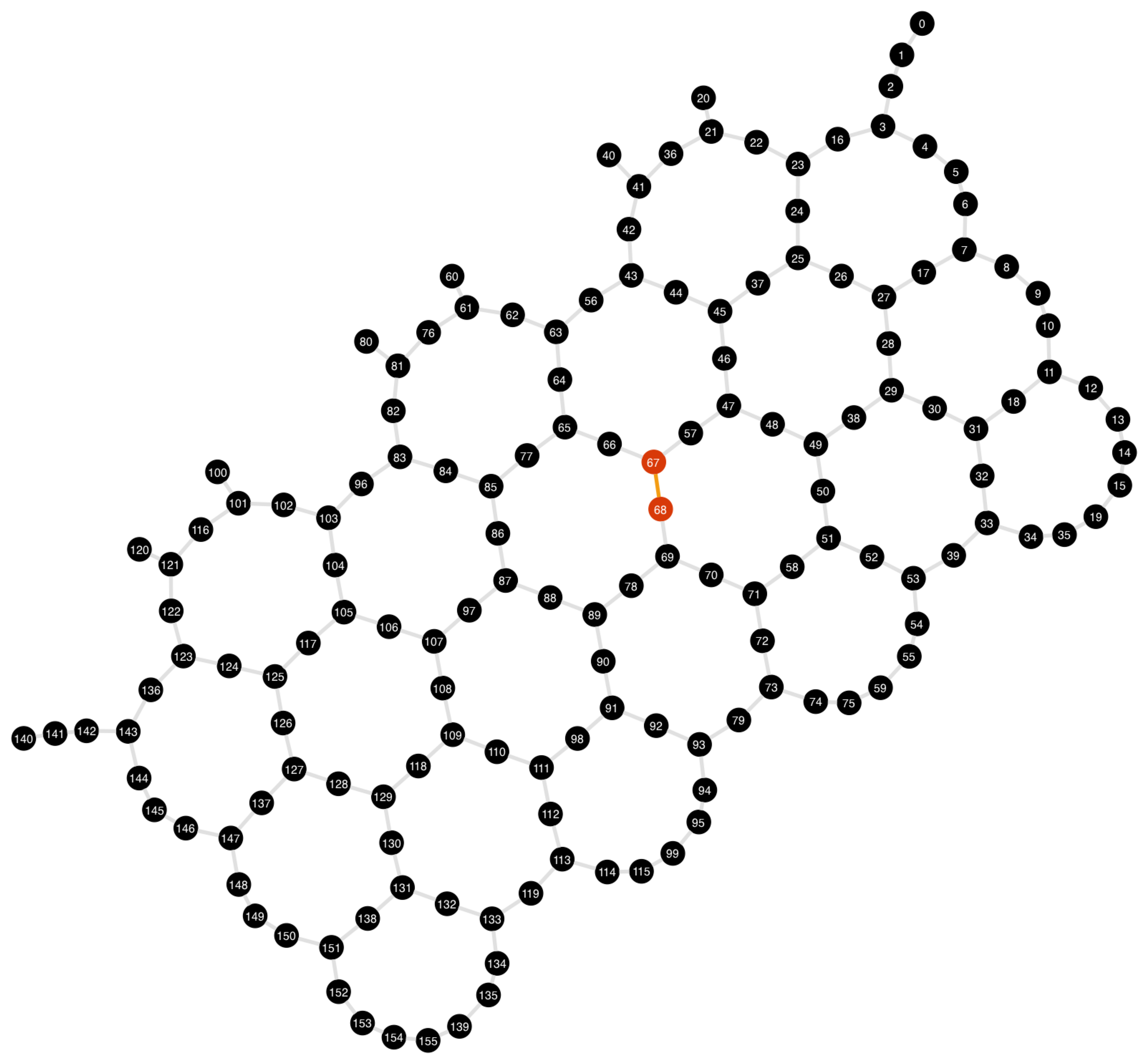}
        \caption{\emph{ibm\_fez} Exp.~1.}
        \label{fig:exp-topology-ibm-1}
    \end{subfigure}\hfill
    \begin{subfigure}[b]{0.23\textwidth}
        \centering
        \includegraphics[width=\linewidth]{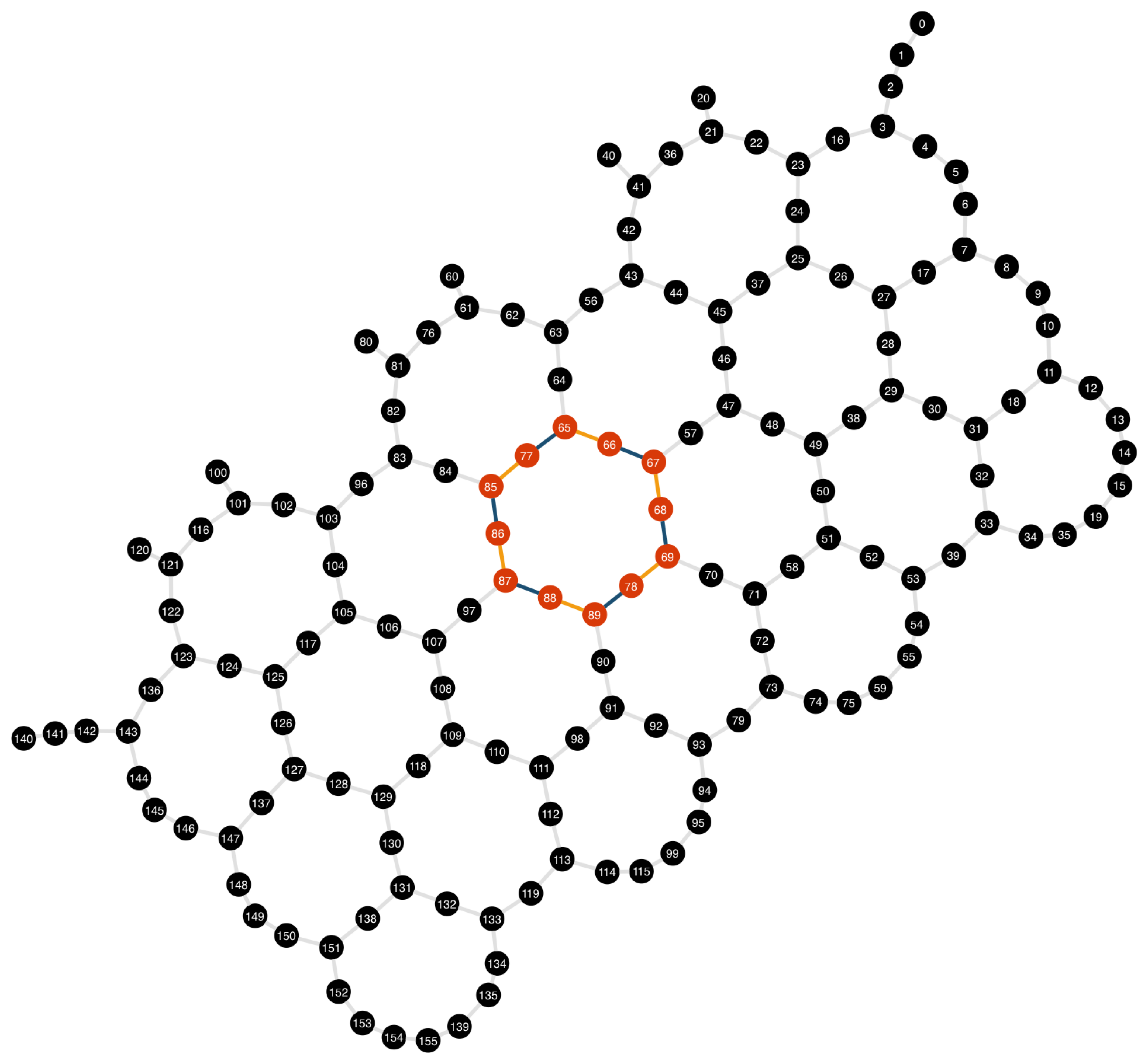}
        \caption{\emph{ibm\_fez} Exp.~2.}
        \label{fig:exp-topology-ibm-2}
    \end{subfigure}\hfill
    \begin{subfigure}[b]{0.23\textwidth}
        \centering
        \includegraphics[width=\linewidth]{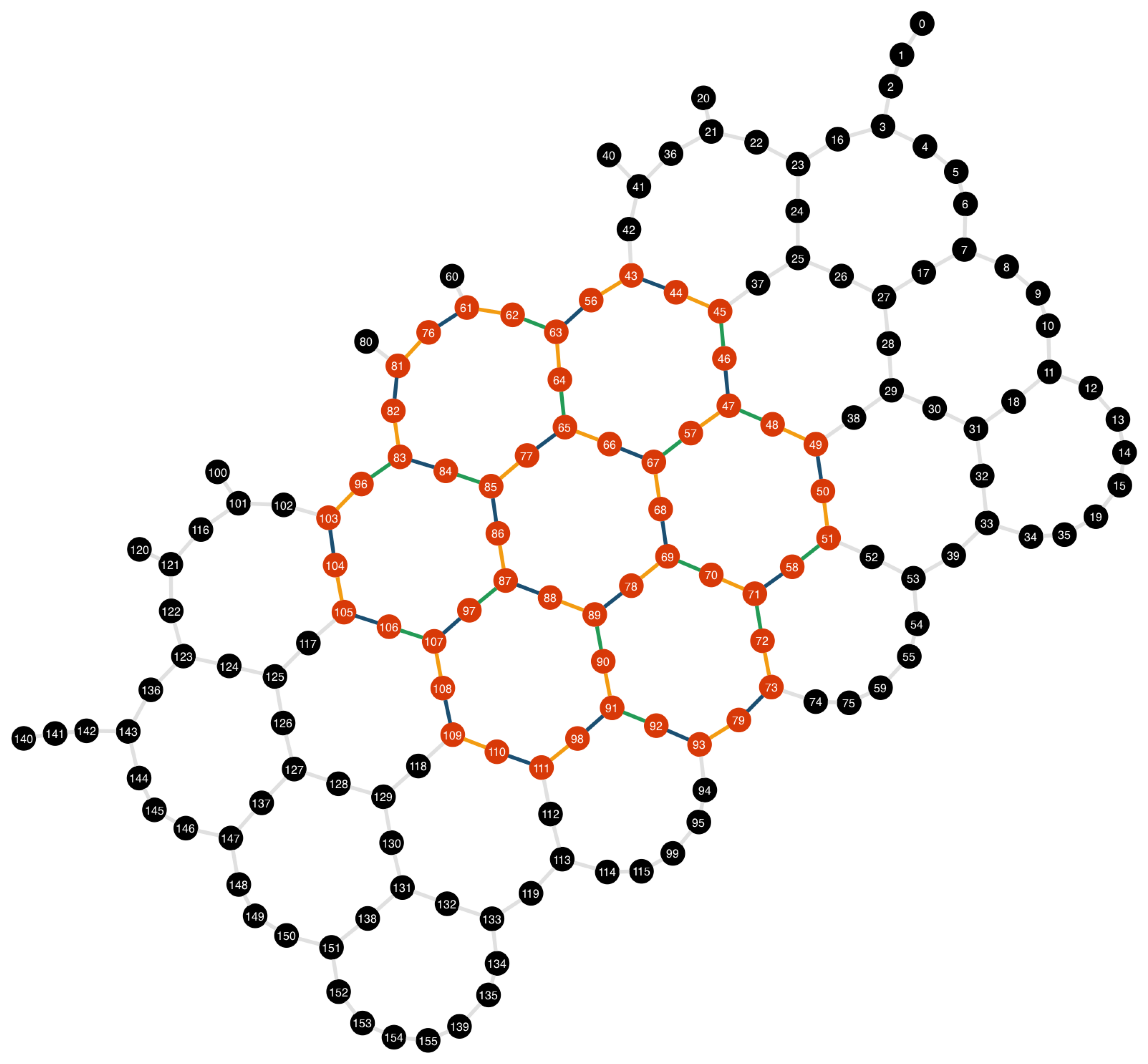}
        \caption{\emph{ibm\_fez} Exp.~3.}
        \label{fig:exp-topology-ibm-3}
    \end{subfigure}\hfill
    \begin{subfigure}[b]{0.23\textwidth}
        \centering
        \includegraphics[width=\linewidth]{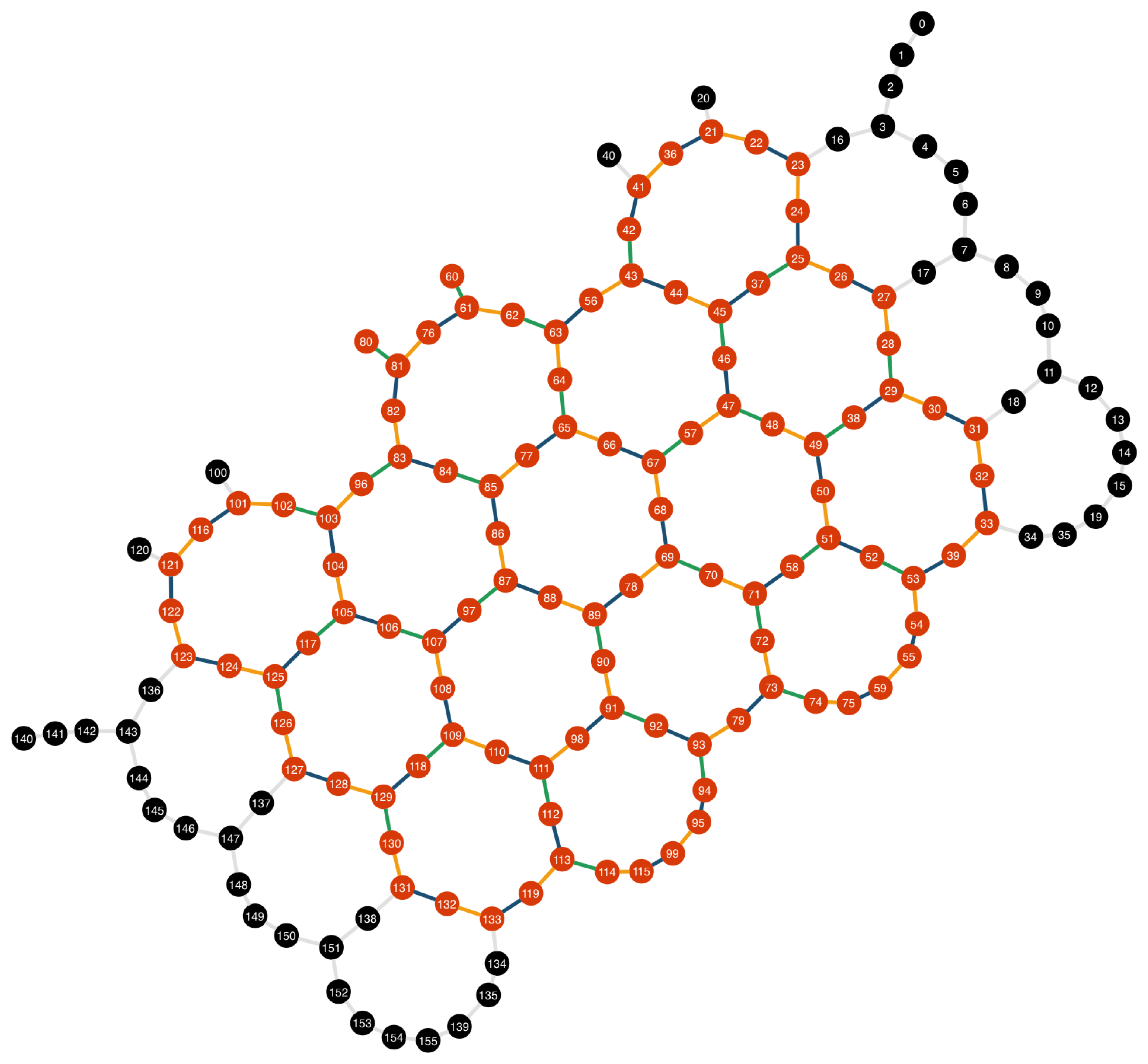}
        \caption{\emph{ibm\_fez} Exp.~4.}
        \label{fig:exp-topology-ibm-4}
    \end{subfigure}

    \caption{Experiment layouts. Top row: IQM \emph{Emerald} (square lattice, $n = 2, 4, 9, 16, 25$). Bottom row: IBM \emph{ibm\_fez} (heavy-hex lattice, $n = 2, 12, 54, 108$). Spectator qubits in black, CZ gate qubits in red; edge sets indicated by like colorm with gray reserved for inactive couplings during experiments.}
    \label{fig:exp-topologies}
\end{figure*}

\subsection{Results}
\label{sec:results}

Detection counts for both devices are summarized in Table~\ref{tab:crosstalk-detection} and in Fig.~\ref{fig:detection-counts}. The NN totals reflect device-specific coupler topology (\emph{Emerald} topology does not support couplers at every lattice position) and rejection criteria described in Section~\ref{sec:ramsey-protocol-and-detection-strategy}, diverging from counts in Table~\ref{tab:experiment-configs}. Detection counts are aggregated at the (spectator, edge-set) level, as described in Section~\ref{sec:ramsey-protocol-and-detection-strategy}.

\begin{table}[ht]
    \centering
    \caption{Detection counts on both devices.}
    \label{tab:crosstalk-detection}
    \begin{tabular}{cc|cc|cc}
    \hline
    \multirow{2}{*}{Device} & \multirow{2}{*}{Exp. ($n$)} & \multicolumn{2}{c|}{Nearest neighbors} & \multicolumn{2}{c}{Non-nearest neighbors} \\
    \cline{3-6}
    & & Detected & Total & Detected & Total \\
    \hline
    \multirow{5}{*}{\emph{Emerald}}
     & 1 (2)  & 1  & 6  & 3  & 45 \\
     & 2 (4)  & 3  & 16  & 7  & 82 \\
     & 3 (9)  & 8  & 26 & 17 & 150 \\
     & 4 (16) & 23 & 39 & 32 & 109 \\
     & 5 (25) & 25 & 43 & 24 & 69 \\
    \hline
    \multirow{4}{*}{\emph{ibm\_fez}}
     & 1 (2)   & 2  & 3   & 7  & 144 \\
     & 2 (12)  & 10 & 12  & 20 & 264 \\
     & 3 (54)  & 13 & 24  & 17 & 267 \\
     & 4 (108) & 20 & 24  & 11 & 114 \\
    \hline
    \end{tabular}
\end{table}

\begin{figure*}[ht]
    \centering
    \begin{subfigure}[b]{0.49\linewidth}
        \includegraphics[width=\linewidth]{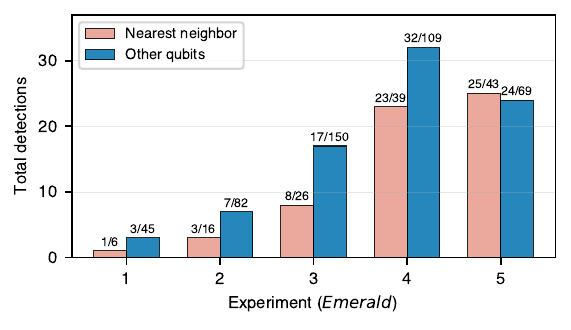}
        \vspace{-20pt}        \caption{\emph{Emerald}, detection threshold $12.85$~kHz.}
        \label{fig:iqm-detection-results}
    \end{subfigure}\hfill
    \begin{subfigure}[b]{0.49\linewidth}
        \includegraphics[width=\linewidth]{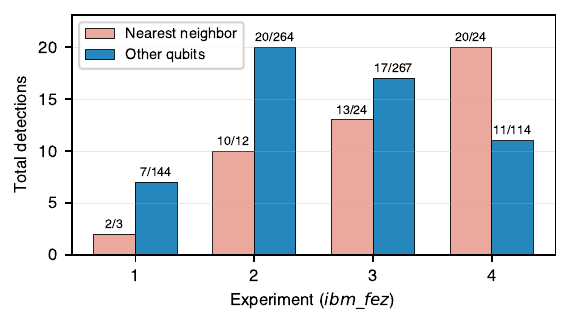}
        \vspace{-20pt}
        \caption{\emph{ibm\_fez}, detection threshold $3.15$~kHz.}
        \label{fig:ibm-detection-results}
    \end{subfigure}
    \caption{Detectable frequency shifts by experiment totaled over all edge sets. Red bars indicate nearest-neighbor detections, blue bars indicate detections on all other spectators.}
    \label{fig:detection-counts}
\end{figure*}

Nearest-neighbor detection peaks at $25/43 \approx 58\%$ on \emph{Emerald} at $n=25$ and $20/24 \approx 83\%$ on \emph{ibm\_fez} at $n=108$, falling short of full detection on the vertex boundary for both devices. We attribute the absence of full saturation to detuning-dependent residual ZZ during active CZs. Tunable couplers suppress ZZ unevenly across spectator-to-gate qubit detunings, so a fraction of NN spectators accumulate sub-threshold phase per CZ. We also observed long-range crosstalk beyond the nearest neighbors in both devices across all experiments. This suggests a side channel Willie may exploit for improved detection.

The two devices exhibit different behavior with respect to long-range crosstalk. On \emph{Emerald}, the non-NN detection rate scales monotonically with~$n$, rising from $\approx 7\%$ at $n=2$ to $\approx 35\%$ at $n=25$. On \emph{ibm\_fez}, the non-NN rate stays in the $5$--$10\%$ range across all four experiments, suggesting consistency across the chip.

To verify that the long-range shifts originate from Alice's CZ gates rather than from coupling between Ramsey spectators, we ran a control experiment on \emph{Emerald}: CZs applied to a single qubit pair $(26, 27)$ while only one spectator (qubit~$44$, four hops away) ran the Ramsey circuit and all other qubits idled. The frequency shift persisted at approximately the same magnitude as when all spectators ran Ramsey circuits in parallel, and the result was reproduced in experiments conducted weeks apart. Full characterization of this long-range coupling for each device is left to future work. We discuss possible mechanisms in the next section.


\section{Conclusion}
\label{sec:conclusion}

We introduced covert quantum computing as a formal security primitive for multi-tenant QPUs and provided its information-theoretic formulation. We adapted the quantum strategy framework~\cite{gutoskiGeneralTheoryQuantum2007,
gutoskiMeasureDistanceQuantum2012} and derived a norm-based covertness criterion, allowing Warden Willie to perform any actions within the confines of physics on his QCUs and quantum memory. We applied our information-theoretic formalism to a model of NISQ superconducting architectures with planar topologies, where the primary modeled crosstalk mechanism is nearest-neighbor spectator phase accumulation from ZZ interactions during two-qubit gates. We developed a vertex-isoperimetric inequality for the heavy-hex lattice and combined it with known results for the square and hexagonal lattices. This yielded a square-root law for superconducting quantum computing: any $n$-qubit circuit admits a placement such that the computation remains covert with $\mathcal{O}(\sqrt{n})$ additional idling qubits along the vertex boundary. Ramsey experiments on IQM's \emph{Emerald} and IBM's \emph{ibm\_fez} confirmed nearest-neighbor crosstalk consistent with this model. However, they also revealed reproducible long-range frequency shifts on spectator qubits outside the boundary. Thus, the square-root law developed here should be interpreted as a guarantee under the nearest-neighbor crosstalk model, rather than as a complete hardware-level guarantee for current devices.

The long-range coupling we observed is itself a contribution, in that it identifies a regime in which the nearest-neighbor noise model underlying Corollary~\ref{cor:square-root-bound-for-covert-computing} fails. One possible mechanism is electromagnetic leakage from an active tunable coupler to the flux line of a distant coupler, partially activating the latter. Full characterization requires chip-level access and experiments over all CZ pairs. The broader implication is that any hardware-level covertness guarantee requires a sufficiently accurate characterization of the system's relevant noise channels, since Willie is assumed to possess this information and enact the optimal detection strategy.

Covert quantum computing is fundamentally a detection problem at the device's physical layer. Unlike attacks that corrupt or extract information at the logical level~\cite{arellanoQubitViseDoubleSidedCrosstalk2025,
campbellSchrodingersToolboxExploring2025,
harperCrosstalkAttacksDefence2025}, error correction cannot protect against an adversary attempting to detect computation. In fact, such actions may reduce Alice's covertness. Hence, covert quantum computing is closely tied to error suppression, which seeks to eliminate unwanted noise in the system.
This first requires characterization and understanding of such noise. As noise determines Willie's detection power, progress on that front reshapes the covertness threat. Furthermore, fully suppressing noise is unrealistic and, therefore, error correction is required for fault tolerance~\cite{bravyiFutureQuantumComputing2022}.
Covertness therefore, remains a meaningful topic beyond the NISQ regime.

Several directions for future work follow naturally beyond noise characterization. Our experimental Ramsey protocol is a fixed prepare-and-measure strategy restricted to a single round ($r = 1$) with no coherent memory coupled to $W^s$, and, therefore, provides only a lower bound on Willie's detection power. Entangled spectators, adaptive operations between rounds, and joint measurement across $W^s$ are likely to tighten the bound, though special consideration needs to be taken to account for Willie's own gate errors. For Alice, it is possible that other choices and tradeoffs exist to remain both reliable and covert. Here we demonstrated that she remains covert for a specific type of noise by using $\mathcal{O}(\sqrt{n})$ qubits and a physical qubit mapping. One such choice extends to the scenario of multiple shots in Section~\ref{sec:unreliable-computation}: she may randomly choose to idle during some shots and compute on others. This is analogous to the sparse-signaling approach in the covertness literature \cite{tahmasbi21signalingcovert}. Furthermore, instead of merely deciding to compute or not, Alice may employ a more sophisticated strategy, adapting to her own observations of Willie's activity as well as other users' activity.
This will contribute to quantum game theory by extending \cite{gutoskiGeneralTheoryQuantum2007, gutoskiMeasureDistanceQuantum2012, chiribellaMemoryEffectsQuantum2008, chiribellaTheoreticalFrameworkQuantum2009} fully into the covertness dimension.
Finally, the landscape of quantum computing architectures is vast. Superconducting, trapped-ion, neutral-atom, photonic, and other platforms encompass their own sub-architectures and implementations. Each of these demands its own analysis of covertness. The work begun here is a first step into exploring that landscape

\section*{Acknowledgment}
The authors benefited from discussions with Saikat Guha and Michael S.~Bullock.

\appendices


\section{Proof of Lemma~\ref{lem:generalize-heavy-isoperimetric}} \label{app:subdivision-lemma-proof}

\begin{proof}
We first prove the inequality \eqref{eq:subdivision-isoperimetric-vertex}
for the vertex boundaries, then describe the modifications required for
\eqref{eq:subdivision-isoperimetric-edge}.

\emph{Achievability.} Let $S^* = S' \cup \{v_{uw} \in S_2 : u \in S' \text{ or } w \in S'\}$. We show $N(S^*) = N(S')$ as subsets of $V(G)$, which gives $|N(S^*)| = |N(S')|$.

$N(S^*) \subseteq N(S')$: Let $x \in N(S^*)$, so $x \notin S^*$ and $x$ has a neighbor in $S^*$. Suppose $x = v_{uw} \in S_2$. Then $v_{uw} \notin S^*$ implies $u, w \notin S'$ by definition of $S^*$ and De Morgan's law. By definition of $S_2$, $v_{uw}$ is adjacent only to $u, w \in V(G')$; since $S^* \cap V(G') = S'$ and $u, w \notin S'$, neither $u$ nor $w$ lies in $S^*$, contradicting $x \in N(S^*)$. Therefore $x \in V(G') \setminus S'$. Since every edge of $G'$ was subdivided, the only neighbors of $x$ in $G$ lie in $S_2$, so $x$ is adjacent in $G$ to some $v_{xw} \in S^*$, which forces $w \in S'$. As $\{x, w\} \in E(G')$, we have $x \in N(S')$.

$N(S^*) \supseteq N(S')$: Let $u \in N(S')$, so $u \notin S'$ and $u$ has a neighbor $w \in S'$ in $G'$. Since $S^* \cap V(G') = S'$, $u \notin S^*$, and $v_{uw} \in S^*$ since $w \in S'$. As $u$ is adjacent to $v_{uw}$ in $G$, we have $u \in N(S^*)$.

\emph{Optimality.} Let $S \subseteq V(G)$ be arbitrary with $S \cap V(G') = S'$. To establish $|N(S)| \geq |N(S')|$, we construct an injection $f: N(S') \to N(S)$. Note that $N(S') \subseteq V(G') \setminus S'$. For each $u \in N(S')$, fix a neighbor $w(u) \in S'$ with $\{u, w(u)\} \in E(G')$, and define
\begin{align}
    f(u) = \begin{cases}
        u & \text{if } v_{u w(u)} \in S, \\
        v_{u w(u)} & \text{if } v_{u w(u)} \notin S.
    \end{cases}
\end{align}
If $v_{u w(u)} \in S$, then $u \notin S$ and $u$ is adjacent to $v_{u w(u)} \in S$, so $f(u) = u \in N(S)$. If $v_{u w(u)} \notin S$, then $v_{u w(u)}$ is adjacent to $w(u) \in S' \subseteq S$, so $f(u) = v_{u w(u)} \in N(S)$.

To see $f$ is injective, consider distinct $u_1, u_2 \in N(S')$:
\begin{enumerate}
    \item If $f(u_1), f(u_2) \in V(G')$, then $f(u_i) = u_i$ and
    $f(u_1) = u_1 \neq u_2 = f(u_2)$.
    \item If $f(u_1), f(u_2) \in S_2$, then
    $f(u_i) = v_{u_i w(u_i)}$. The edges
    $\{u_1, w(u_1)\}$ and $\{u_2, w(u_2)\}$ are distinct: equality as
    unordered pairs would require either $u_1 = u_2$ (ruled out by
    assumption) or $u_1 = w(u_2)$ (ruled out since
    $u_1 \notin S'$ while $w(u_2) \in S'$). Hence
    $f(u_1) \neq f(u_2)$.
    \item If $f(u_1) \in V(G')$ and $f(u_2) \in S_2$ (or vice versa),
    then $f(u_1) \neq f(u_2)$ since $V(G') \cap S_2 = \emptyset$.
\end{enumerate}

For \eqref{eq:subdivision-isoperimetric-edge}, the same argument applies with two changes. The optimal set is $S^* = S' \cup \{v_{uw} \in S_2 : u \in S' \text{ and } w \in S'\}$, where both endpoints must be in $S'$. Under this condition, each boundary edge of $S'$ contributes exactly one boundary edge of $S^*$ in $G$, giving $|\mathcal{E}(S^*)| = |\mathcal{E}(S')|$. For
\emph{optimality}, the injection sends a boundary edge $\{u, w\}$ of $S'$, labeled such that $u \in S'$ and $w \notin S'$, to $\{v_{uw}, w\}$ if $v_{uw} \in S$ and to $\{u, v_{uw}\}$ otherwise.
\end{proof}

\clearpage
\bibliographystyle{IEEEtran}
\bibliography{./papers.bib}

\end{document}